\documentclass[11pt]{article}

\usepackage[letterpaper, margin=1in]{geometry}
\usepackage{bm, amsmath, amssymb, amsthm, amsfonts}
\usepackage{thmtools,thm-restate}
\usepackage{cite}
\usepackage{framed}
\usepackage[framemethod=tikz]{mdframed}
\usepackage{appendix}
\usepackage{graphicx}
\usepackage{color}
\usepackage{varwidth}
\usepackage{wrapfig}
\usepackage{algorithm}
\usepackage[noend]{algpseudocode}
\usepackage[textsize=tiny]{todonotes}

\newcommand{\logstar}{\ensuremath{\log^*}}

\newcommand{\wmin}{\ensuremath{w_{\textrm{min}}}}
\newcommand{\wmax}{\ensuremath{w_{\textrm{max}}}}	

\usepackage{enumitem}
\setitemize{noitemsep,topsep=3pt,parsep=3pt,partopsep=3pt}

\usepackage{xspace}
\usepackage{xifthen}

\newcommand*{\myproofname}{\textit{Proof of Claim 1:}}

\definecolor{darkgreen}{rgb}{0,0.5,0}
\definecolor{darkblue}{rgb}{0,0,0.8}
\usepackage{hyperref}
\hypersetup{
    unicode=false,          
    colorlinks=true,        
    linkcolor=darkblue,          
    citecolor=darkgreen,        
    filecolor=magenta,      
    urlcolor=cyan           
}
\RequirePackage[]{silence}
\usepackage[nameinlink,capitalize]{cleveref}

\newtheorem{theorem}{Theorem}[section]
\newtheorem{lemma}[theorem]{Lemma}

\theoremstyle{definition}
\newtheorem{definition}{Definition}[section]
\theoremstyle{remark}
\newtheorem*{remark}{Remark}

\newcommand{\ignore}[1]{}

\algnewcommand\algorithmicswitch{\textbf{switch}}
\algnewcommand\algorithmiccase{\textbf{case}}

\algdef{SE}[SWITCH]{Switch}{EndSwitch}[1]{\algorithmicswitch\ #1\ \algorithmicdo}{\algorithmicend\ \algorithmicswitch}%
\algdef{SE}[CASE]{Case}{EndCase}[1]{\algorithmiccase\ #1}{\algorithmicend\ \algorithmiccase}%
\algtext*{EndSwitch}%
\algtext*{EndCase}%

\newcommand{\CONGEST}{\ensuremath{\mathsf{CONGEST}}\xspace}
\newcommand{\LOCAL}{\ensuremath{\mathsf{LOCAL}}\xspace}

\newcommand{\CREWPRAM}{\ensuremath{\mathsf{CREW}~\mathsf{PRAM}}\xspace}

\newcommand{\eps}{\varepsilon}
\newcommand{\poly}{\operatorname{\mathrm{poly}}}

\newcommand{\set}[1]{\left\{#1\right\}}

\DeclareMathOperator{\polylog}{\poly\log}

\newcommand{\myparagraph}[1]{

\medskip

\noindent \textbf{\boldmath #1}}

\newcommand{\bigO}{O}

\newcommand{\dcup}{\dot\cup}
\newcommand{\approxmatching}[1]{\approxmatchingx{#1}{\Delta}}
\newcommand{\approxmatchingx}[2]{T_{\text{WM}}\ensuremath{(n, #2, #1)}}
\newcommand{\match}[2]{T_{\text{HM}}\ensuremath{(n, #2, #1)}}
\newcommand{\gmax}{g_{\textrm{max}}}

\newcommand{\actualmaxdegree}{\Lambda}

\newcommand{\hide}[1]{}

\newcommand{\FullOrShort}{short}

\ifthenelse{\equal{\FullOrShort}{full}}{
	
  \newcommand{\fullOnly}[1]{#1}
  \newcommand{\shortOnly}[1]{}

}{
  
  \usepackage{times}

  \newcommand{\shortOnly}[1]{#1}
  \newcommand{\fullOnly}[1]{}

}

\newenvironment{DenseEnumerate}[1][\theenumi.]
{\begin{list}{#1}{\usecounter{enumi}
\itemsep 0pt \parsep 0pt \leftmargin 5mm \labelwidth 5mm \parskip 0pt \topsep
0pt}}
{\end{list}}

\begin{document}

\begin{flushleft}

\vspace*{0.8cm}
{\huge\bf Deterministic Distributed Edge-Coloring \\ with Fewer Colors\par}
\vspace{1.0cm}
\end{flushleft}

\newcommand{\auth}[3]{\textbf{#1}$\,\,\,\cdot\,\,\,$#2$\,\,\,\cdot\,\,\,$#3\par\medskip}

\auth{Mohsen Ghaffari}
{ETH Zurich}
{ghaffari@inf.ethz.ch}
\auth{Fabian Kuhn\footnote{Supported by ERC Grant No.\ 336495 (ACDC).}}
{University of Freiburg}
{kuhn@cs.uni-freiburg.de}
\auth{Yannic Maus\footnotemark[1]}
{University of Freiburg}
{yannic.maus@cs.uni-freiburg.de}
\auth{Jara Uitto\footnotemark[1]}
{ETH Zurich and University of Freiburg}
{jara.uitto@inf.ethz.ch}

\vspace{1cm}

\begin{abstract}
We present a deterministic distributed algorithm, in the $\LOCAL$ model, that computes a $(1+o(1))\Delta$-edge-coloring in polylogarithmic-time, so long as the maximum degree $\Delta=\tilde{\Omega}(\log n)$. For smaller $\Delta$, we give a polylogarithmic-time $3\Delta/2$-edge-coloring. These are the first deterministic algorithms to go below the natural barrier of $2\Delta-1$ colors, and they improve significantly on the recent polylogarithmic-time $(2\Delta-1)(1+o(1))$-edge-coloring of Ghaffari and Su [SODA'17] and the $(2\Delta-1)$-edge-coloring of Fischer, Ghaffari, and Kuhn [FOCS'17], positively answering the main open question of the latter. The key technical ingredient of our algorithm is a simple and novel gradual packing of judiciously chosen near-maximum matchings, each of which becomes one of the color classes.

\end{abstract}
\setcounter{page}{0}
\thispagestyle{empty}
\newpage

\section{Introduction \& Related Work}

Edge-coloring is one of the four classic problems in the subarea of \emph{local distributed algorithms}---along with vertex-coloring, maximal independent set (MIS), and maximal matching---which have been studied extensively since the 1980s; see, e.g., the distributed graph coloring book of Barenboim and Elkin~\cite{barenboimelkin_book}. While there are rather satisfactory randomized algorithms, finding efficient deterministic algorithms consitute some of the most well-known and central open questions of the area; see, e.g., the first five open problems of \cite{barenboimelkin_book}. In this paper, we present efficient deterministic edge-coloring algorithms that use considerably fewer colors, for the first time going below the barrier of $2\Delta-1$ colors, where $\Delta$ denotes the maximum degree.

\subsection{Background}
\myparagraph{\boldmath $\LOCAL$ Model.}
We work with the standard $\LOCAL$ model of distributed computing, usually attributed to Linial\cite{linial1987LOCAL}: the network is abstracted as an $n$-node undirected graph $G=(V, E)$, and each node is labeled with a unique $O(\log n)$-bit identifier. Communication happens in synchronous \emph{message passing} rounds, where in each round each node can send a message to each of its neighbors. At the end of the algorithm, each node should output its own part of the solution, e.g., the colors of its incident edges in the edge-coloring problem. The time complexity of an algorithm is the number of synchronous rounds.

\paragraph{History, The Journey to Deterministic \boldmath$(2\Delta-1)$-Edge-Coloring:}
Any graph with maximum degree at most $\Delta$ admits a $(2\Delta-1)$-edge-coloring, by a trivial sequential greedy algorithm. Moreover, a very simple randomized distributed algorithm, following from Luby's 1986 MIS algorithm\cite{luby86, alon86}, can compute such a coloring in $O(\log n)$ rounds. This made computing a $(2\Delta-1)$-edge-coloring a natural first-target for deterministic distributed algorithms. However, in contrast to randomized algorithms, computing such a coloring deterministically and efficiently (in polylogarithmic time) remained an open problem for about 30 years (see, e.g., Open Problem 11.4 in the distributed graph coloring book\cite{barenboimelkin_book}). Until very recently, the best known round complexity for $(2\Delta-1)$-edge-coloring was $2^{O(\sqrt{\log n})}$ by using an algorithm of \cite{panconesi95}.

A brief recap of the concrete steps of progress is as follows: Since Linial's pioneering 1987 paper~\cite{linial1987LOCAL} and for many years, $O(\Delta^2)$-colors was the best known palette size for polylogarithmic-time algorithms. Two intermediate steps of progress were polylogarithmic-time algorithms of Czygrinow et al.\cite{czygrinow2001coloring} for $O(\Delta\log n)$-edge-coloring and that of Barenboim and Elkin\cite{Barenboim:edge-coloring} for $\Delta\cdot 2^{O(\log \Delta/\log\log \Delta)}$-edge-coloring. More recently, a significant leap was made by Ghaffari and Su\cite{ghaffari17} who presented a $(2+o(1))\Delta$-edge-coloring in polylogarithmic-time. Finally, the question was settled very recently when Fischer, Ghaffari and Kuhn~\cite{FOCS17} presented a $O(\log^7 \Delta \log n)$-round deterministic distributed algorithm for $(2\Delta-1)$-edge-coloring. For $((2+o(1))\Delta)$-edge-coloring, a simpler and faster algorithm with time complexity \mbox{$\tilde{\bigO}\big(\log^2\Delta \cdot \log n)$} was given recently by Ghaffari et al.\cite{Splitting17}.\footnote{Throughout the paper, we use $\tilde{\bigO(\cdot)}$ to hide factors that are polynomial in $\bigO(\log\log n)$.}

In the special case where $\Delta$ is assumed to be small and we investigate the complexity's dependency on $\Delta$, the following results were presented for $(2\Delta-1)$-edge-coloring: Panconesi and Rizzi~\cite{panconesi2001some} gave an $\bigO(\Delta+\logstar n)$-round algorithm. This was improved by Barenboim \cite{barenboim15} to $\bigO(\Delta^{3/4}\log \Delta+\logstar n)$ and then further to $\bigO(\Delta^{1/2}\log^{5/2} \Delta+\logstar n)$ by Fraigniaud, Heinrich and Kosowski \cite{fraigniaud16}. The last two results solve the more general problem of $(\Delta+1)$-vertex coloring problem. 

\paragraph{Going Below \boldmath$2\Delta-1$ Colors:}
By a classic result of Vizing\cite{vizing1964,Bollobas98a}, any graph admits a $(\Delta+1)$-edge-coloring. However, the proofs of existence and the algorithms for computing such a coloring are fundamentally different from those for $(2\Delta-1)$-edge-coloring: for instance, for $(2\Delta-1)$ coloring, any partial coloring can be extended to a full coloring of all edges, without altering the already colored edges; moreover, this is the smallest number of colors for which such an extension is always possible. Furthermore, all known proofs of $(\Delta+1)$-edge-coloring rely on very non-local arguments, in the sense that coloring one more edge may depend on the color of (and may end up recoloring) edges that are even $\Theta(n)$-hops away. 

Similarly, $(2\Delta-1)$ colors have been viewed as a natural barrier for deterministic edge-coloring algorithms and no such algorithm was known that uses (considerably) less than $2\Delta-1$ (formally less than $2\Delta-2$ colors), even for low-degree graphs. This remained a known open question since the 1990's, see e.g., Panconesi and Srinivasan\cite{Panconesi1992}. Note that in contrast, at least for low-degree graphs---e.g., when $\Delta\leq \polylog(n)$---computing a $(2\Delta-1)$-edge-coloring is easy, and was known since Linial's paper\cite{linial1987LOCAL}. 

Given that the $(2\Delta-1)$-edge-coloring problem is solved, and noting this intrinsic change in the nature of the problem once we go below $2\Delta-1$ colors, an intriguing question that remained open, and was explicitly asked by Fischer et al.\cite{FOCS17}, is
\begin{center}
``\emph{How close can we get to this [Vizing's edge-coloring], while remaining in polylogarithmic-time?}"
\end{center}

\subsection{Our Contributions}

In this paper, we almost settle the above question. Concretely, we get a deterministic distributed coloring that is within a $1+o(1)$ factor of Vizing's bound, so long as $\Delta=\omega(\log n)$, and for smaller degrees, we show how to get $3\Delta/2$-edge-coloring. Below, we present the formal statement of the results---the precise constants in the exponents are not optimized and can be found in the technical section.

\begin{restatable}{theorem}{thmendedgecoloring}
\label{thm:finalColoring}
There is a constant $c>0$ such that for every $\eps>0$ there exist deterministic distributed algorithms that color the edges of any $n$-node graph that has maximum degree at most $\Delta$ with 
\begin{itemize}
\item $\Delta+\eps\Delta$ colors if $\Delta\geq c\cdot\eps^{-1}\cdot\log\eps^{-1}\cdot\log n$ in $O\big(\eps^{-9}\log^{O(1)}n\big)$ rounds 
\item $3\Delta/2$ colors for all $\Delta$ in $O\big(\Delta^{9}\log^{O(1)}n\big)$ rounds.
\end{itemize}
\end{restatable}

We remark that the complexities stated above depend in a black-box manner on the complexity of computing hypergraph maximal matchings and ($1-\eps$)-approximations of weighted maximum matching. Above, we have stated the bound based on the current state of the art, which are presented in the concurrent work~\cite{newHypergraphMatching}; our technical sections make this dependency explicit. If one prefers not to depend on this simultaneous work\cite{newHypergraphMatching}, we could use the solutions provided in \cite{FOCS17}, and still obtain a polylogarithmic-time $(1+\eps)\Delta$-edge coloring for arbitrarily small constant $\eps>0$.

Another remark is that one can choose the $\eps$ parameter in \Cref{thm:finalColoring} to be quite small, and as a result, can get a coloring with $\Delta+\poly(\log n)$ colors, as we formalize in the next corollary. This is interesting because the known randomized methods cannot go below the bound of $\Delta+\sqrt{\Delta}$ colors, which is a natural barrier for those methods\cite{Pettie17}, rooted in the standard deviations of the random nibble step.

\begin{restatable}{corollary}{thmfirstedgecoloring}
 \label{thm:edgeColoringLastStep}
There is a deterministic distributed algorithm that colors the edges of any $n$-node graph  with maximum degree at most $\Delta$ with  
	$\Delta + \bigO\left( \log n \log \left( 2 + \frac{\Delta}{\log n} \right) \right)$ colors in $\bigO \left( \Delta^{9} + \log^{\bigO(1)} n \right)$ rounds. 
\end{restatable}

Finally, we note that in a recent simultaneous work, Ghaffari, Harris, and Kuhn\cite{newHypergraphMatching} present a generic method for derandomizing local distributed algorithms. As we explain in \Cref{sec:derandomized}, combining their result with the randomized edge coloring algorithm of \cite{Pettie17}, we can obtain an alternative method for proving a slightly weaker version of the first part of \Cref{thm:finalColoring}. We still believe that the algorithm that we present in the main body of this paper has a number of advantages: (1) the lower bound on $\Delta$ for which this method works is better, (2) it is much simpler, cleaner, and more comprehensible in comparison with the algorithm that comes out of the ``automated" application of derandomization (via a certain method of conditional expectations) atop the non-trivially complex randomized algorithm of\cite{Pettie17}, (3) the computations and communications in this main algorithm are simpler and more efficient, e.g., the algorithm fits in the \CONGEST model, modulo the part of ($1-\eps$)-approximation of maximum-weight matchings, which itself can presumably be improved in the future to work in the \CONGEST model.

\subsection{Our Method in a Nutshell}
Here, we provide a brief outline of our $(1+o(1))\Delta$-edge-coloring algorithm, which is made of two components. The latter of which is the main novelty of this paper. The first component is a splitting algorithm, borrowed from \cite{ghaffari17, Splitting17}, which partitions the edge-set of the graph into roughly $(1+o(1))\Delta/d$ disjoint sets, each of which forms a spanning subgraph with max-degree $d=\tilde{O}(\log n)$. This effectively means that all that remains to be solved is $(1+o(1))\Delta$-edge-coloring for graphs of maximum degree $\tilde{O}(\log n)$.
  
We now discuss our $(1+o(1))\Delta$-edge-coloring for low-degree graphs, which is our main technical contribution. This algorithm is based on an iterative packing of judiciously crafted matchings, each of which will be one color class. To provide some intuition for this, let us consider a simple (though not ideal) algorithm for the easier objective of $(2\Delta-1)$-edge-coloring: simply, for $2\Delta-1$ iterations, in each iteration $i$, compute a maximal matching, color all of its edges with color $i$, and remove them from the graph. This will color all the edges because in each iteration, for each remaining edge $e$, either edge $e$ or at least \emph{one} of its incident edges gets colored. To get to a $(1+o(1))\Delta$-edge-coloring, our hope is that the matching computed in each iteration is more ``expansive" and for each remaining edge $e$ (formally, some appropriate relaxation of this ``for each" guarantee), the computed matching removes at least \emph{two} of its incident edges. This way, the neighborhood of $e$ drops at a rate of two edges per iteration and thus, the $2\Delta-1$ edges in this neighborhood are expeced to get exhausted in about $\Delta$ iterations, i.e., after about $\Delta$ colors. 

To gain more intuition, let us consider another hypothetical scenario, a bipartite $\Delta$-regular graph. This graph has a perfect matching \cite{alon2003simple, vizing1964}; if we could compute such a matching distributedly, we could remove it, remain with a bipartite $(\Delta-1)$-regular graph, and repeat, eventually ending with a $\Delta$-edge-coloring. However, we cannot compute a perfect matching efficiently in the distributed setting (this problem may need $\Omega(n)$ rounds). Our hope would be to use instead some near-maximum matching, but this creates some irregularity in the degrees, which can grow as we continue adding more and more matchings. Our algorithm follows a similar outline but has to be much more careful in managing these irregularities.

Concretely, the algorithm's core is as follows: For simplicity, assume that $\Delta$ is lower bounded by some sufficiently large polylogarithmic bound; the actual algorithm which optimizes this lower bound will need more care. For some $T=\Theta(\log n)$ iterations, in each iteration $t\in \{1, 2, \dots, T\}$, we find a near-maximum matching with the following special property: For each $i\in \{1, \dots, t\}$, for the set $S_i$ of all nodes of degree at least $\Delta-i$, this single matching is incident on a $(1-o(1))$-fraction of $S_i$. We call such a structure a \emph{pervasive matching}. One can see that for each $i$, such a matching exists. The core technical challenge in our algorithm will be to find one matching that satisfies this property for all $i$ \emph{simultaneously}. Our algorithm for computing a pervasive matching is abstracted by \Cref{cor:degreepriorities} and its proof appears in \Cref{sec:matching}. After $T$ steps, as we will show, this will ensure that the maximum degree has reduced by about $T(1-o(1))$. Hence, repeating this idea will result in a $(1+o(1))\Delta$-edge-coloring, in roughly $(1+o(1))\Delta$ iterations. Once the remaining degree $\Delta'$ becomes too small, say $\Delta'=o(\Delta)$, we cannot find such a nice matching, but we clean up that case more coarsely, by just packing maximal matchings; we thus use about $2\Delta'$ colors at that point, but this $2$-factor is negligible overall and we still get a $(1+o(1))\Delta$-edge-coloring.

Our $3\Delta/2$-edge-coloring algorithm also has a similar iterative structure, but now each iteration extracts some special subgraph, called $(3)$-graph, instead of a matching. This structure is such that (1) we are able to color its edges using $3$ colors, and (2) removing it reduces the maximum degree upper bound by an additive $2$. The extraction of $(3)$-graphs itself relies on the computation of a sequence of five maximal and maximum matchings of subgraphs of $G$; the latter is possible locally only thanks to the extra special properties of the bipartite graphs in consideration in those intermediate steps of the algorithm. 

\subsection{Other Related Work}
\myparagraph{Randomized Edge-Coloring Algorithms:}  The classic randomized MIS algorithm of Luby \cite{luby86,alon86} leads to a randomized $(2\Delta-1)$-edge-coloring algorithm with runtime $\bigO(\log n)$. This was improved to $O(\log \Delta+2^{\sqrt{\log\log n}})$ by Barenboim et al. \cite{barenboim12} and to $O(\sqrt{\log \Delta}+2^{\sqrt{\log\log n}})$ by Harris, Schneider and Su \cite{hsinhao_coloring}, both of which work also for the more general problem of $(\Delta+1)$-vertex-coloring. Fischer et al.\cite{FOCS17} gave a randomized $(2\Delta-1)$-edge-coloring algorithm with complexity $O(\log^8\log n)$.

In contrast to the deterministic setting there are quite a few randomized algorithms which use fewer than $2\Delta-1$ colors. The first such result was by Panconesi and Srinivasan \cite{Panconesi1992}; their result was later improved Dubhashi, Grable, Panconesi, to a $(1+\eps)\Delta$-edge-coloring algorithm in polylogarithmic time with the restriction that $\Delta=\Omega(\log^{1+\Omega(1)}n)$. 
The time complexity for the same number of colors was improved Elkin, Pettie und Su  to $O\big(\logstar\Delta \cdot \max\big(1,\frac{\log n}{\Delta^{1-o(1)}}\big)\big)$, \cite{Elkin15}. In \cite{Pettie17} Chang, He, Li, Pettie, and Uitto recently showed how to edge-color a graph with $\Delta+\tilde{O}(\sqrt{\Delta})$ colors in $O(\log \Delta \cdot T_{LLL})$ where $T_{LLL}$ is the complexity of a permissive version
of the constructive Lov\'{a}sz local lemma. 
We also refer to \cite{Pettie17} for a more detailed survey of randomized edge-coloring algorithms.

\myparagraph{Lower bounds for Edge-Coloring:} The celebrated $\Omega(\logstar(n))$ round complexity lower bound of Linial's\cite{linial1987LOCAL} still remains the only lower bound known for $(2\Delta-1)$-edge-coloring \cite{linial1987LOCAL}. Recently, Chang et al.\cite{Pettie17} proved that $(2\Delta-2)$-edge-coloring has lower bounds of $\Omega(\log_{\Delta} n)$ for deterministic and $\Omega(\log_{\Delta}\log n)$ for randomized algorithms. Moreover, one natural way to compute $(\Delta+1)$-edge-colorings is to extend partial colorings by iteratively recoloring edges along an `augmenting paths'. Chang et al.\cite{Pettie17} showed that with this `recoloring-along-a-path' approach one might have to recolor nodes along paths of length $\Omega(\Delta \log n)$ to color a single additional edge.

\myparagraph{Deterministic Edge-Coloring Algorithms for Low-Arboricity Graphs:}
%
In \cite{Elkin17} Barenboim, Elkin and Maimon gave a polylogarithmic-time determinisitc ($\Delta+o(\Delta)$)-edge-coloring for graph with arboricity $a \leq \Delta^{1-\eps}$, for a constant $\eps>0$. This was then improved by Fischer et al.\cite{FOCS17} to $(\Delta+(2+\eps)a-1)$-colors. Chang et al.\cite{Pettie17} present a $\Theta(\log_{\Delta} n)$-round $\Delta$-edge-coloring algorithm for trees.

\myparagraph{Distributed Maximum Matching Approximation:} Computing almost optimal matchings is at the core of our algorithms.  The standard approach to $(1 - \eps)$-approximate the maximum matching problem is by Hopcroft and Karp~\cite{Hopcroft73}.  The main obstacle to transfer the framework to a distributed setting is to find maximal sets of disjoint augmenting paths of length $\bigO(1/\eps)$.  Czygrinow and Ha\'{n}\'{c}kowiak gave a deterministic algorithm that runs in time $\log^{\bigO(1/\eps)} n$~\cite{czygrinow2003distributed}, which was recently improved to $\log^{\bigO(\log( 1/\eps ))} n$ by Fischer et al.~\cite{FOCS17}.  The most recent result by Ghaffari et al.~\cite{newHypergraphMatching} takes the $\eps$-dependency out of the exponent, yielding a $\poly(\log n / \eps)$-time algorithm.  On the randomized side, Lotker et al.\cite{LPP08} developed an $\bigO(\log n / \eps^2)$-time $(1 - \eps)$-approximation for maximum matching (also see \cite{LPP08} for additional work on randomized distributed matching approximation).



\section{Distributed Edge-Coloring Algorithm through Iterative Matchings}
\label{sec:iterativecoloring}
In this section, we first present the core of our matching algorithm which is especially efficient in graphs with intermediate degrees---e.g., polylogarithmic---and then, in \Cref{sec:edgeColoringwithSplitting}, we explain how to lift this algorithm to higher degrees, with the help of splitting techniques\cite{ghaffari17, Splitting17}. Before proceeding to the algorithm we present some notations that we use to express our complexities.

\paragraph{Notation for the Complexity of Some Black-Box Subroutines} Our algorithm makes use of two subroutines in a black-box manner, and thus the final complexity of our algorithm depends on the complexity of the (best known) method for them. To express this complexity explicitly and illustrate the dependencies, we use some notation. We define $\match{\ell}{\Gamma}$ to be the runtime of a hypergraph maximal matching algorithm with $n$ nodes, maximum degree at most $\Gamma$ and rank at most $\ell$. The best known published bound is from \cite{FOCS17}, roughly being $\log^{O(\log r)} \Delta \cdot \log n$, and an improvement to $\poly(r\log(\Delta n))$ was recently provided by\cite{newHypergraphMatching}. Moreover, we define $\approxmatching{\eps}$ for the runtime of a maximal  $(1-\eps)$-approximate weighted maximum matching on a simple graph with $n$ nodes and maximum degree at most $\Delta$; we provide an upper bound on the latter in \Cref{lemma:distrweightedmatching}, which is at most $\match{1/\eps}{\Delta^{O(1/\eps)}}\cdot \poly(\log n/\eps)$. 
 
\paragraph{Algorithm Outline.} Our edge-coloring algorithm constructs the color classes iteratively. In each iteration, the algorithm
constructs a matching that \emph{hits}---i.e., is incident on--- a large fraction of all the nodes. In
order to guarantee progress on all the nodes, when constructing the
next matching, we make sure that the nodes that have been hit by fewer
previous matchings are given more priority to be hit by the next
matching.

More precisely, in step $t$ of the algorithm,  we compute a single matching $M$ that simultaneously hits almost all of the nodes of degree at least $\Delta-i$ for each $i=1,\ldots,t$ (cf. \Cref{cor:degreepriorities}, detailed proof in \Cref{sec:matching}). 
Then, we color all edges of $M$ with an unused color and remove $M$ from the graph. 
If the initial maximum degree is large enough and we repeat this process for  $T=\Theta(\log n)$ steps, the maximum degree will reduce by $\approx(1-\eps)T$ (cf. \Cref{lemma:edgecoloringprogress}). Repeating this will eventually color the whole graph with few colors.
The runtime of this algorithm is inherently at least $\Omega(\Delta)$ as it computes the matchings sequentially and each matching corresponds to a single color class, i.e., it is only efficient for graphs with polylogarithmic degree. In \Cref{sec:edgeColoringwithSplitting} we show how degree splittings can be used to transform it into an efficient algorithm for all degrees. 
The following lemma is proved in detail in \Cref{sec:matching}.
\begin{lemma}(\textbf{The Pervasive Matching Lemma})\label{cor:degreepriorities}
  Consider a graph $G=(V,E)$ with maximum degree at most $\Delta$ and an
  integer  $t\geq 0$. For every $\eps>0$, there is an
  $\bigO\left( t/\eps +\approxmatching{\eps/2}\right)$-round distributed algorithm
  to compute a maximal matching $M$ such that for every $i\leq t$, $M$ hits
  a $(1-\eps)\cdot \frac{\Delta-i}{
\actualmaxdegree+1}$-fraction of all the
  nodes of degree at least $\Delta-i$, where $\actualmaxdegree\leq \Delta$ exactly equals the maximum degree of $G$.
\end{lemma}

We continue to prove that when starting with a sufficiently large
maximum degree $\Delta$, we can reduce the maximum degree at a rate  close to one by iteratively computing matchings with \Cref{cor:degreepriorities}. For readability, we provide proof sketches explaining the main ideas and full proofs of the following two lemmas.

\begin{lemma}\label{lemma:edgecoloringprogress} 
  Assume that we have an $n$-node graph $G=(V,E)$ with maximum degree at most 
  $\Delta\geq \frac{2\log n}{\eps}$.
	Then there is an algorithm that partially edge colors the graph with 
  $T:=\big\lceil\frac{\log n}{4e\eps}\big\rceil$ colors such that the maximum degree of the uncolored graph is 
  at most $\Delta - (1-4e\eps)T$.
	
	For each of the $T$ colors, the algorithm has round complexity
	 $\bigO\left(  T/\eps +\approxmatching{\eps} \right)$ . 
\end{lemma}
\myparagraph{Algorithm:} We start with all edges uncolored. 
In step $t\in \{1,\ldots,T\}$, we choose the set of edges that are to be colored with color $t$. 
 Let $G_t$ be the graph induced by the set of uncolored edges
after the first $t-1$ steps and let $\Delta_t$ be the maximum degree
of graph $G_t$. 
 With these definitions at hand, step $t$ of the algorithm uses a single invocation of \Cref{cor:degreepriorities} on $G_t$ to
compute a maximal matching $M$ such that for each
$\delta\in\set{\max\set{1, \Delta - T }, \dots, \Delta}$, $M$ covers at least a
$(1-\eps)\cdot\frac{\delta}{\Delta_t+1}$-fraction of all nodes with degree at least $\delta$. We assign color $t$ to all edges in $M$ and remove them from the graph.

We first provide a proof sketch for the Lemma to indicate the main ideas of the proof. The full proof follows right afterwards.
\begin{proof}[Proof Sketch:]
The runtime of one step of the algorithm  follows from the runtime of \Cref{cor:degreepriorities} with $t=T$.
  In the rest of the proof we need to upper bound the maximum degree of the uncolored graph after $T$ steps. 
	
	It is essential that the matching in a single step of the algorithm hits at least a \mbox{$(1-\Theta(\eps))$}-fraction of the large degree nodes. If we assumed a stronger requirement on $\Delta$, e.g., \mbox{$\Delta\geq 2 \eps^{-2} \cdot \log n$}, it is easy to see that 
	$$\frac{\Delta-T+1}{\Delta_t+1} \geq \frac{\Delta-T+1}{\Delta+1}=1-\frac{T}{\Delta+1} \geq 1-\eps$$
holds, that is, the computed matching would hit at least a $(1-\eps)^2$-fraction of the nodes of degree at least $\Delta-i$ for all $i\leq T$. With the weaker assumption on $\Delta$ one has to carefully track $\Delta_t$ to (at least) show that for all $i\leq t$ we have $\frac{\Delta-t+i}{\Delta_t+1}\geq \frac{\Delta-t}{\Delta_t+1}\geq 1-\eps$ (subclaim in the full proof).  
Thus, the matching in step $t$ hits at least a $(1-\eps)^2$-fraction, that is, at least a  $(1-2\eps)$-fraction, of the nodes of degree at least $\Delta-t+i$ for $i\leq t$. Then, by an induction on the number of rounds, one can show that for all $i\leq t$ the number of nodes with degree at least $\Delta-t+i$ after $t$ steps of the algorithm is less than $\binom{t}{i}\cdot (2\eps)^i \cdot n$~.
	
  Finally, to prove the main claim of \Cref{lemma:edgecoloringprogress}, we need to show that  the number of nodes with degree at least $\Delta-T+4e\eps T+1$ after round $T$ is smaller than one. This holds because with $i=4e\eps T+1$ we have
  \begin{align*}
	\binom{t}{i}\cdot (2\eps)^i \cdot n\leq
   \left(\frac{e T}{i}\right)^i(2\eps)^i\cdot n =\left(\frac{2e\eps T}{i}\right)^i\cdot n
    \ \stackrel{(i>4e\eps T)}{<}\
    \left(\frac{1}{2}\right)^{4e\eps T}\cdot n\ \stackrel{T=\big\lceil\frac{\log n}{4e\eps}\big\rceil}{\leq}\ 1. &\qedhere
  \end{align*}
\end{proof}
We now give a formal proof of the statement.
\begin{proof}
The runtime of one step of the algorithm  follows from the runtime of \Cref{cor:degreepriorities} with $t=T$.

  In the rest of the proof we upper bound the maximum degree of the uncolored graph after $T$ steps. 
	For  $\eps\geq 1/(4e)$ the statement holds as $\Delta - (1-4e\eps)T\geq \Delta$. So from now on we assume that $\eps<1/(4e)$. 
	
	For each $t\in\set{0,\dots,T}$ and for all
  $i\in\set{0,\dots,t+1}$, define $K(t,i)$ to be the number of nodes
  of degree at least $\Delta-t+i$ after the first $t$ steps (i.e., in
  graph $G_{t+1}$). Note that $K(t,t+1)=0$ for
  all $t\geq 0$. 
	
	We first show by induction on $t$ that for all
  $t\in\set{0,\dots,T}$ and $i\in\set{0,\dots,t}$, the following holds
  \begin{equation}\label{eq:nodesbehindbyi}
    K(t,i) \leq \binom{t}{i}\cdot (2\eps)^i \cdot n.
  \end{equation}
	
  \noindent\textit{Induction Base:} Note that for $i=0$ (and any $t\geq 0$)
   \Cref{eq:nodesbehindbyi} is trivial as it states that $K(t,i)\leq n$.
		Thus, \Cref{eq:nodesbehindbyi} also holds for  $t=0$ because $i=0$ is the only possible value with $i\leq t=0$. 
	
	\noindent\textit{Induction Step:} Consider some step $t\geq 1$ and $i\leq t$. We first prove the following subclaim which is necessary to show that the computed matching hits a fraction that is very close to one of the nodes of degree at least $\Delta-t+i$. 
	
	\medskip
	
	\noindent\textbf{Subclaim :} $\frac{\Delta-t+i}{\Delta_t+1}\geq 1-\eps$.
	\begin{proof}
Let $I$ be the smallest integer such that
  $K(t-1,I)=0$. 
	By the definition of $I$, 
 the maximum degree $\Delta_t$ after $t-1$ steps is at most
  $\Delta_t\leq \Delta-(t-1)+I-1=\Delta-t+I$.   
	Then the ratio
  $\frac{\Delta-t+i}{\Delta_t+1}$  simplifies to
  \begin{align*}
  \frac{\Delta-t+i}{\Delta_t+1}\ \geq\ \frac{\Delta-t+i}{\Delta-t+I+1}\
  \stackrel{(i\geq 0)}{\geq}\ \frac{\Delta-t}{\Delta-t+I+1}\ =\
  1- \frac{I+1}{\Delta-t+I+1}\ \stackrel{I\geq 0}{\geq}\ 1 - \frac{I+1}{\Delta-t+1}.
  \end{align*}
	
	To show the claim we now prove that  $\frac{I+1}{\Delta-t+1} \leq \eps$ holds.
Using $\Delta\geq \frac{2\log n}{\eps}$  and $t\leq T\leq \big\lceil\frac{\log n}{4 e \eps}\big\rceil$ we can first lower bound the denominator in the  term.
	\begin{align}
	\label{eqn:boundonDt}
	\Delta-t+1\geq \frac{2\log n}{\eps}- \left\lceil \frac{\log n}{4 e \eps} \right\rceil+1\geq\frac{2\log n}{\eps}-\frac{\log n}{4 e \eps}  = \frac{8e-1}{4e\eps}\cdot \log n~. 
	\end{align}
	Next, we upper bound the nominator in the term, i.e., we upper bound $I+1$. Let $j\geq \log n$. If $j> (t-1)$ then  $K(t-1,j)=0$ by definition. Otherwise, by the induction hypothesis and using $\binom{t-1}{j}< \left( \frac{e(t-1)}{j} \right)^j$ as well as $t-1\leq T-1 \leq \lfloor\frac{\log n}{4e\eps}\rfloor\leq \frac{\log n}{4e\eps}$, we have 
  \[
  K(t-1,j) \leq \binom{t-1}{j}\cdot (2\eps)^j\cdot n 
  \ <\  \left(\frac{2\eps e(t-1)}{j}\right)^j \cdot n\ \leq\ \left(\frac{2\eps e\cdot \frac{\log n}{4e\eps}}{j}\right)^j \cdot n\ \leq\ 1.
  \]
  Thus, $K(t-1,j)=0$ holds for $j\geq \log n$. Hence, we obtain that $I+1\leq \lceil \log n \rceil +1\leq \log n +2$.

	With $I+1\leq \log n +2$ and the lower bound on $\Delta-t+1$ from \Cref{eqn:boundonDt} we obtain
  \begin{align*}
    \frac{I+1}{\Delta-t+1} 
     \ \leq\   \eps\cdot \frac{4e}{8e-1} \cdot
             \frac{\log n   + 2}{\log n}  
				\stackrel{}{\ \leq\ } \eps~. 
  \end{align*}
	\phantom\qedhere
	The last inequality follows with $n\geq 8$. \hfill $\blacksquare$
\end{proof}

	Now, we can proceed with proving the induction step. Recall, that $t$ and $i$ are fixed. Then there are two
  different types of nodes of degree at least $\Delta-t+i$ after step
  $t$: $(1)$  Nodes of degree at least $\Delta-t+i+1$ after step $t-1$. $(2)$ Nodes of degree $\Delta-t+i$ after step
  $t-1$ that are not hit by the matching computed in step $t$. We now upper bound the number of nodes of each type.
	\begin{enumerate}[label=(\arabic*)]
	\item By the definition there are $K(t-1,i)$ nodes of type $(1)$.
  \item We show that there are at most $2\eps\cdot K(t-1,i-1)$ nodes of type $(2)$.
	
  Note that $\Delta-t+i$ lies in the range of $\delta$ when applying \Cref{cor:degreepriorities} as $\Delta-t+i \geq \Delta - T$.
  Thus, by \Cref{cor:degreepriorities}, of the nodes of degree at least $\Delta-t+i$ after step $t-1$, at least a
  $(1-\eps)\cdot\min\set{1,\frac{\Delta-t+i}{\Delta_t+1}}$-fraction is
  hit by the matching in step $t$. 
 The subclaim $\frac{\Delta-t+i}{\Delta_t+1}\geq 1-\eps$  implies that this fraction is at least a $(1-2\eps)$-fraction.
	
	The number of nodes that are not hit is therefore at most
	$(1-(1-2\eps))K(t-1,i-1) {\leq} 2\eps\cdot K(t-1,i-1).$
	\end{enumerate}
	
  Due to the bounds in $(1)$ and $(2)$ we obtain $K(t, i)  \leq K(t-1, i) + 2\eps\cdot K(t-1, i-1)$ and plugging in the induction hypothesis leads to \Cref{eq:nodesbehindbyi} as follows.
  \begin{align*}
    K(t, i) 
		\leq \binom{t-1}{i}\cdot (2\eps)^i\cdot n +
    2\eps\cdot \binom{t-1}{i-1}\cdot (2\eps)^{i-1}\cdot n
    \ =\ \binom{t}{i}\cdot (2\eps)^i\cdot n.
  \end{align*}
	This finishes the induction step.

  Finally, to prove the main claim of \Cref{lemma:edgecoloringprogress}, we show that $\Delta_{T+1}\leq \Delta - T + 4e\eps
  T$, or equivalently that for $i>4e\eps T$, we have
  $K(T,i)<1$.
	 Using \Cref{eq:nodesbehindbyi} and $\binom{T}{i}<(eT/i)^i$, we get
  \[
    K(T,i) < \left(\frac{e T}{i}\right)^i(2\eps)^i\cdot n =\left(\frac{2e\eps T}{i}\right)^i\cdot n
    \ \stackrel{(i>4e\eps T)}{<}\
    \left(\frac{1}{2}\right)^{4e\eps T}\cdot n \leq\ 1 \ .
  \]
  The last inequality follows because
  $T=\big\lceil\frac{\log n}{4e\eps}\big\rceil$.
\end{proof}

\begin{restatable}{lemma}{lemmaepsEdgeColoring}
 \label{thm:epsedgeColoring}
There is a constant $c>0$ such that for any $\eps>0$ there exists deterministic distributed algorithm that colors the edges of any $n$-node graph  with maximum degree at most $\Delta\geq c\cdot \eps^{-1}\log \eps^{-1} \log n$ with  
	$\Delta+\eps\Delta$ colors and has round complexity
$\bigO\left(\Delta \cdot (\log n/\eps^2 +\approxmatching{\eps/2} ) \right)$ .
\end{restatable}
We first provide a proof sketch for a weaker result; the proof of the full result follows afterwards.
\begin{proof}[Proof Sketch]
In the full proof of the lemma one has to apply \Cref{lemma:edgecoloringprogress} with increasing values for $\eps$ and perform a  careful analysis of the number of used colors. 
In this proof sketch we prove a slightly weaker result, i.e., we show that we can color the graph with  $(1+O(\eps))\Delta$ colors if $\Delta\geq \Delta':=2\eps^{-2}\log n$ .

To obtain the desired result we apply \Cref{lemma:edgecoloringprogress} with $\eps$ until the current upper bound on the maximum degree that is guaranteed by the lemma falls below $\Delta'$. Then perform a clean-up step in which the remaining uncolored graph is colored with $2\Delta'-1$ colors, e.g., by computing further $2\Delta'-1$ maximal matchings with \Cref{cor:degreepriorities}.
A single application of  \Cref{lemma:edgecoloringprogress} uses at most $T=\big\lceil\frac{\log n}{4e\eps}\big\rceil$ colors, reduces the maximum degree of the graph by $(1-4e\eps)T$, and there are at most \mbox{$K=\frac{\Delta-\Delta'}{(1-4e\eps)T}>0$} applications of the lemma. 
Thus, the total number of colors is bounded by 
$$K\cdot T +2\Delta'-1 \leq \frac{\Delta-\Delta'}{(1-4e\eps)}+2\Delta'\leq  (1+24 \eps)(\Delta- \Delta')+2\Delta'=\Delta+O(\eps)\Delta~.$$
Computing a single color class with \Cref{cor:degreepriorities} needs $\bigO\left(  T/\eps +\approxmatching{\eps} \right)$ rounds. 
Thus, the runtime for the $(1+O(\eps))\Delta$ color classes is bound by $O\left(\Delta\cdot \left(  \log n/\eps^2 +\approxmatching{\eps/2}\right)\right)$ .
\end{proof}
We continue with the full proof.
\begin{proof}
Define the constant $c:=360$. 
For $i=0,\ldots, l$ define $\eps':=\eps/120$, $\eps_i:=2^i\eps'$, the threshold degrees \mbox{$\Delta_i:=2\log n /\eps_i$} and $T_i= \left\lceil \frac{\log n}{4e\eps_i} \right\rceil$ where $l\leq \log 1/(4e\eps)'-1\leq \log 1/\eps+3$ is the smallest $l$ such that  $\eps_{l+1}\geq 1/(4e)$ holds. 
Then, the algorithm consists of phases $i=0,1,\ldots, l$ and a clean-up step. 
 
In phase $i$ we apply \Cref{lemma:edgecoloringprogress} with $\eps_i$ until the current upper bound of the maximum degree, which is guaranteed by the lemma, falls below $\Delta_i$. Then we continue to the next phase. Note that no node has to know the actual current maximum degree of the uncolored graph for this process. After phase $l$ the maximum degree of the remaining graph will be  at most $\Delta_l$ and we use the clean-up step to  color the  edges of the remaining graph with $2\Delta_l\leq 32e\cdot \log n\leq 96 \cdot \log n$ colors, e.g., by computing $2\Delta_l-1$ further  maximal matchings with \Cref{cor:degreepriorities}.

We now  upper bound the number of used colors. For that purpose define $\Delta_{-1}:=\Delta$. In phase \mbox{$i=0,1,\ldots, l$} a single application of  \Cref{lemma:edgecoloringprogress} reduces the maximum degree of the graph by $(1-4e\eps)T_i$~, we use at most $T_i$ colors in each application of the lemma and there are at most $K_i=\frac{\Delta_{i-1}-\Delta_i}{(1-4e\eps_i)T_i}>0$ applications of the lemma.
Thus, the total number of colors used in phase $i$ is upper bounded by
$$K_i\cdot T_i  = \frac{\Delta_{i-1}-\Delta_i}{(1-4e\eps_i)}\leq  (1+24 \eps_i)(\Delta_{i-1}- \Delta_i) \ .$$

During the $l+1$ phases and the clean-up step, we use 
\begin{align*}
\sum_{i=0}^{l}(1+24 \eps_i)(\Delta_{i-1}-\Delta_i)+2\Delta_l &\ =\  (1+24\eps')\Delta+\sum_{i=1}^l24\eps_i\Delta_i +\Delta_l \\
& \stackrel{(*)}{\ \leq\ }  (1+24\eps')\Delta +(48l+96)\log n\\
& \ \leq\  (1+24\eps')\Delta +(48(\log 1/\eps+3)+96)\log n\\
& \ \leq\ (1+24\eps')\Delta+48\log 1/\eps\log n + 240\log n\\
& \stackrel{(**)}{\ \leq\ } \Delta+\frac{3}{15}\eps \Delta + \frac{2}{15}\eps \Delta+ \frac{10}{15}\eps \Delta = (1 + \eps)\Delta \ , 
\end{align*}
that is, $\leq(1+\eps)\Delta$ colors in total. At $(*)$ we used $\eps_i\Delta_i= 2\log n$ and $\Delta_l\leq 96\log n$. At $(**)$ we used $24\eps'\leq 1/5\cdot \eps$  and $360 \log n \log \eps^{-1}\leq \eps \cdot \Delta$ .

\paragraph{Runtime:}
Computing a single color class in phase $i$ needs 
\[
	\bigO\left( T_i/\eps +\approxmatching{\eps_i} \right)=\bigO\left(  T/\eps +\approxmatching{\eps} \right) 
\]
rounds. 
As we compute $(1+\eps)\Delta$ color classes the total runtime is upper bounded by
\begin{align*} 
O\left(\Delta\cdot \left(  \log n/\eps^2 +\approxmatching{\eps/2}\right)\right) \ . & \qedhere
\end{align*} 
\end{proof}

\Cref{thm:epsedgeColoring,lemma:edgecoloringprogress} are simpler to prove if we tolerate a larger dependency on $\epsilon$ in the lower bound for $\Delta$. However, with that increased dependency we would not only lose  in the runtime but, more importantly, the number of colors could not go below $\Delta+\Omega\big(\sqrt{\Delta}\big)$ (as in \Cref{sec:goingBelow}) regardless of the time that we spend.


\subsection[Proof of The Pervasive Matching Lemma]{Proof of The Pervasive Matching Lemma (\Cref{cor:degreepriorities})}
\label{sec:matching}
In this section we prove \Cref{cor:degreepriorities}. First, in \Cref{sec:approxmatching} we prove that a simple consequence of Vizing's edge coloring theorem (cf. \Cref{lemma:largematchingsexits}) shows that for all $i=1,\ldots$  there exists a maximum matching that hits at least a $\frac{\Delta-i}{\actualmaxdegree+1}$-fraction of the nodes with degree at least $\Delta-i$, where $\actualmaxdegree$ is the exact maximum degree of the graph. 
Then a distributed implementation  of the \CREWPRAM weighted maximum matching approximation algorithm by Hougardy and Vinkemeier~\cite{Hougardy2006} (cf. \Cref{lemma:distrweightedmatching}) can be used to compute a matching $M_i$ that hits a 
$(1-\eps)\frac{\Delta-i}{\actualmaxdegree+1}$-fraction of the nodes with degree at least $\Delta-i$ (cf. \Cref{lemma:distrlocalmatching}).
In  \Cref{sec:combiningMatching} (\Cref{lemma:combinematchings}) we first show how two matchings $M_i$ and $M_j$ where $M_i$ hits a $(1-\eps)\frac{\Delta-i}{\actualmaxdegree+1}$-fraction of the nodes of degree at least $\Delta-i$ and $M_j$ hits a  $(1-\eps)\frac{\Delta-j}{\actualmaxdegree+1}$-fraction of the nodes of degree at least $\Delta-j$ can be combined into a single matching that has both properties.
Then in \Cref{lemma:prioritymatching,cor:degreepriorities} we compute matchings $M_i$ for each $i=1,\ldots,\Theta(\log n)$ and then use the \Cref{lemma:combinematchings} to iteratively combine them into the single matching $M$ while maintaining their properties. 

\subsubsection{A Matching that Hits Most Nodes of One Target Node Set}
\label{sec:approxmatching}
A consequence of Vizing's edge coloring theorem \cite{vizing1964,Bollobas98a} shows that large maximum matchings exist. We use it to show that our maximum matching approximations hit enough nodes.
\begin{lemma}\label{lemma:largematchingsexits}
 Given a graph $G=(V,E)$ and a node set $S\subseteq V$. If the maximum degree of $G$ is $\actualmaxdegree$ and all nodes in $S$ have degree
  at least $\delta_S$, there exists a matching of $G$
  that hits at least $\frac{\delta_S}{\actualmaxdegree+1}\cdot |S|$ of the nodes in $S$.
\end{lemma}
\begin{proof}
By Vizing's theorem, the graph $G$ has an edge-coloring with
  $\Delta+1$ colors. Hence, the graph $G$ contains $\Delta+1$ disjoint
  matchings such that each node in $S$ is hit by at least
  $\delta_S$ of the matchings. On average, the $\Delta+1$
  matchings therefore hit at least $\frac{\delta_S}{\Delta+1}\cdot|S|$ nodes in $S$. Consequently, one of the $\Delta+1$ matchings has to
  hit at least $\frac{\delta_S}{\Delta+1}\cdot |S|$ nodes in $S$.
	\end{proof}

To prove the following lemma we essentially provide a distributed version of the weighted maximum matching approximation algorithm by Hougardy and Vinkemeier~\cite{Hougardy2006} that runs in the \CREWPRAM model. Its formal proof emphasizing the differences to the proof in \cite{Hougardy2006} can be found in \cref{sec:weightedmatching}.
As combination of \Cref{lemma:distrweightedmatching,lemma:largematchingsexits}, one obtains a distributed algorithm that finds a matching that approximately satisfies the properties of the matching guaranteed to exist by \Cref{lemma:largematchingsexits}.

\begin{lemma}\label{lemma:distrweightedmatching}
  Let $G=(V,E)$ be an $n$-node graph with positive edge weights $w:E\to \mathbb{R}^+$ and let $\wmin$ and $\wmax$ denote the minimum and the maximum edge weight, respectively.
  If $\wmax / \wmin = n^{\bigO(1)}$, then, for every $\eps>0$, a maximal $(1-\eps)$-approximate weighted
   matching can be computed in time 	
	\[
		\approxmatching{\eps}=
	\bigO\left(1/\eps^2+ (1/\eps) \cdot \match{1/\eps}{\Delta^{O(1/\eps)}}\cdot \log n +\log^3 n\right) \ .
	\]
\end{lemma}

\begin{lemma}\label{lemma:distrlocalmatching}
  Given a graph $G=(V,E)$ and a node set $S\subseteq V$. If
  all nodes in $G$ have degree at most $\Delta$ and all nodes in $S$
  have degree at least $\delta_S$, then, for every $\eps>0$, there is
  a distributed algorithm with time complexity
  $\approxmatching{\eps}$ that computes a matching that hits a
  $(1-\eps)\cdot\frac{\delta_S}{\actualmaxdegree+1}$-fraction of the nodes in $S$, where $\actualmaxdegree\leq \Delta$ exactly equals the maximum degree of $G$.
\end{lemma}
\begin{proof}
  The problem of finding a matching $M$ of $G$ that hits as many nodes
  as possible of $S$ can be formulated as a maximum weighted matching
  problem as follows. We define the weight of an edge $\set{u,v}\in E$
  of $G$ as $w(\set{u,v}):=|S\cap\set{u,v}|$. Hence, the total weight
  of a matching $M$ is equal to the number of nodes in $S$ hit
  by $M$. Thus, \Cref{lemma:largematchingsexits} implies that a
  $(1-\eps)$-approximation of this maximum weighted matching problem
  gives a matching that hits at least a
  $(1-\eps)\cdot\frac{\delta_S}{\actualmaxdegree+1}$-fraction of the nodes in $S$. By
  \Cref{lemma:distrweightedmatching} such an approximate weighted
  matching can be computed in time $\approxmatching{\eps}$. 
	Note that when computing the matching, edges of weight $0$ can be ignored.
\end{proof}

\subsubsection{Combining Matchings while Maintaining their Properties}
\label{sec:combiningMatching}

The core application (cf. the proof of \Cref{lemma:prioritymatching}) of the next lemma will be a combination of a matching $M_a$ that hits a large fraction of all nodes with a matching $M_b$ that hits a large fraction of a subset $S\subseteq V$  into a single matching $M_c$ that hits a large fraction of the nodes in $S$ and a large fraction of all nodes.

\begin{lemma}\label{lemma:combinematchings}
  Given a graph $G=(V,E)$ and a node set $S\subseteq V$. Assume that
  we are given two matchings $M_a$ and $M_b$ of $G$ such that matching
  $M_b$ hits at least $s\leq |S|$ nodes of $S$. 
  Then for every
  $k\geq 1$, in $O(k)$ rounds, we can compute a matching $M_c$ such that
	\begin{enumerate}[label=(\roman*)]
	\itemsep 0pt \parsep 0pt \leftmargin 5mm \labelwidth 5mm \parskip 0pt \topsep 0pt
	\item $|M_c|\geq |M_a|$
	\item $V(M_a) - (V(M_c) \cup S)\leq S\cap V(M_c) - V(M_a)$, i.e., for every node outside $S$ that is matched by $M_a$ and not matched by $M_c$ there is a node inside $S$ that is matched by $M_c$ and not matched by $M_a$. 
	\item $M_c$
  matches at least $\big(1-\frac{1}{k}\big)s$ nodes of $S$. 
	\end{enumerate}
\end{lemma}
\begin{proof}	We start with some notation, then present a simple algorithm that computes the matching $M_c$, and finally show that the matching has the desired properties.
		At first restrict $M_b$ only to edges that have at least one node in $S$.	
Throughout, we denote the edges in the symmetric difference of $M_b$ and $M_a$ as \emph{blue} and \emph{green} edges, i.e., denote  edges in $M_b - M_a$ as blue edges and edges in $M_a - M_b$ as green edges. 
  Let $S_a$, $S_b$ and $S_c$ denote the nodes of $S$ that are matched by $M_a$, $M_b$ and $M_c$, respectively.
	
The subgraph induced by the blue and green edges contains alternating (in blue/green) paths and cycles as every node has at most one incident blue and green edge. 

An (undirected) path $p$ is called \emph{alternating} if its edges are blue and green in an alternating manner. An alternating path of length at least one is called a \emph{maximal alternating path} if it cannot be extended to a longer alternating path. A maximal alternating path in which at least one of its \emph{endpoint-edges} is blue is called a \emph{maximal augmenting path}.
Note that the graph induced by blue and green edges consists of alternating paths and cycles and that all nodes on an alternating cycle are matched by $M_a$ and $M_b$. Further note that all maximal alternating paths are node disjoint.

\myparagraph{\boldmath Algorithm to compute $M_c$:}
\begin{DenseEnumerate}
\item Add all edges of $M_a$ to $M_c$.
\item 
For every maximal augmenting path $p$ of length at most $4k$ such that at least one of its endpoints lies in  $S_b - S_a$ remove its green edges from $M_c$, i.e., remove $p\cap M_a$ from $M_c$, and add its blue edges to $M_c$, i.e, add $p\cap M_b$ to $M_c$.
\end{DenseEnumerate}

The algorithm is well defined and as it only imposes changes on maximal augmenting paths of length at most $4k$ it can be implemented by a distributed algorithm in $O(k)$ rounds. 

\myparagraph{Analysis:}
Throughout the analysis $M_c$ and the set $S_c$ denote the sets after the execution of the algorithm.

\medskip
\noindent \textit{Matching Property:} The only edges that are changed are on (short) maximal augmenting paths. For a single maximal augmenting path, removing the green edges from the matching and adding the blue edges to the matching does not destroy the matching property as augmenting paths are alternating and the paths are maximal. 
Furthermore, all resulting maximal augmenting paths are node disjoint and thus, the operations on all maximal augmenting paths can be performed in parallel and still result in a matching. We continue with proving properties $i)-iii)$. 
\begin{DenseEnumerate}
\item[i)] As each maximal augmenting path begins with a blue edge  the matching size does not decrease by switching the edges on an augmenting path.

Here we want to emphasize that the total number of matching edges in $M_c$ does not increase by handling an augmenting path if the path ends with a green edge. 
However, the term 'augmenting path' is still valid in the sense that it increases the number of nodes within $S$ that are matched. 

\item[ii)] We prove this property by finding one distinct node in $W_1:=S_c - S_a=S\cap V(M_c) - V(M_a)$ for every node in $W_2:=V(M_a) - (V(M_c) \cup S)$. As we add all edges from $M_a$ to $M_c$ in step one of the algorithm a node $v$ can only be in $W_2$ if it became unmatched in step two, i.e., if $v$ was an endpoint of a maximal augmenting path and had an incident green edge before handling the path. 
However, then the other endpoint $v'$ of the path had an incident blue edge (here we use that only maximal augmenting paths with one endpoint in $S_b - S_a$ perform changes) and $v'$ is contained in $W_1$. Thus, a simple induction on the number of handled maximal augmenting paths shows the result.

\item[iii)] Every node $v\in S_b - S_c$, i.e., a node in $S$ which is matched by $M_b$ but is not matched by $M_c$ at the end of the algorithm, has an incident blue edge and no incident green edge. 
Thus, it lies at the start of a maximal alternating path $p$ that starts with a blue edge. 
As $v$ is unmatched in $M_c$, the path $p$ cannot be a maximal augmenting path of length at most $4k$.
Because $p$ starts with a blue edge this implies that the length of $p$ is at least $4k+1$ and $p$ contains  at least $2k$ nodes in $S_b$. As all alternating paths are node disjoint and any path can have at most two endpoints this implies $|S_b - S_c|\leq |S_b| \cdot 2/2k = |S_b|/k$. Thus, we have
$|S_c|\geq |S_b\cap S_c| = |S_b|-|S_b - S_c| \geq (1-1/k)|S_b|\geq (1-1/k)s$~.\qedhere
\end{DenseEnumerate}
\end{proof}
\label{sec:mergingMatching}

The next lemma proves the main distributed matching result that is
needed to iteratively compute a good edge-coloring. For a given decreasing chain of node sets $U_1\supseteq U_2\supseteq \ldots U_t$ the lemma computes a single matching that hits a large fraction of the nodes in each $U_i, i=1,\ldots,t$. In its proof we use \Cref{lemma:distrweightedmatching} to compute $t$  matchings $M_1,\ldots,M_t$ such that $M_i$ hits a large fraction of $U_i$.  Then we use \Cref{lemma:combinematchings} to iteratively combine the matchings into a single matching while maintaining their properties. 

\begin{lemma}\label{lemma:prioritymatching}
  Let $G=(V,E)$ be a graph with maximum degree at most
  $\Delta$. Further, assume that we are given $t$ disjoint nodes sets
  $V_1,V_2,\dots,V_t$ such that the minimum node degree of the nodes
  in $V_i$ is at least $\delta_i$, where $\delta_i\leq \delta_{i+1}$
  for  all $i\in\{1,\ldots,t-1\}$.
	
	For every $\eps>0$, there is a
  distributed $\bigO\left(t/\eps+ \approxmatching{\eps/2}\right)$-round algorithm that
  computes a maximal matching $M$ such that for every $i\in \set{1,\dots,t}$,
  $M$ hits an $(1-\eps)\cdot\frac{\delta_i}{\actualmaxdegree+1}$-fraction of the
  nodes in $$U_i := \bigcup_{j=i}^t V_j~,$$
	where $\actualmaxdegree\leq \Delta$ exactly equals the maximum degree of $G$.
\end{lemma}
\begin{proof}
  First, for each $i\in \set{1,\dots,t}$ we compute a matching $M_i$
  that hits a $(1-\eps/2)\frac{\delta_i}{\actualmaxdegree+1}$-fraction of all nodes in
  $U_i$ by using \Cref{lemma:distrlocalmatching} and the fact that
  all nodes in $U_i$ have degree at least $\delta_i$. All of these matchings can
  be computed in parallel. Now, the
  matching $M$ is constructed inductively by applying
  \Cref{lemma:combinematchings} $t$ times.
  Formally, we iteratively construct matchings $M_1',\dots,M_t'$ such
  that for each $i\in \set{1,\dots,t}$ and each $j\leq i$ matching
  $M_i'$ hits at least a
  $(1-\eps)\cdot\frac{\delta_j}{\actualmaxdegree+1}$-fraction of all nodes in
  $U_j$. This proves the lemma as the matching $M_t'$
  satisfies all required conditions, except for the maximality. To obtain matching $M$ we extend $M_t'$ to a maximal matching in time $\bigO(\log^3 n)$ with the algorithm by Fischer~\cite{Fischer17}.
	
	\myparagraph{\boldmath Inductive construction of $M_1',\dots,M_t'$:} Set
  $M_1':=M_1$. For $i>1$, set $M_i':=M_c$ where $M_c$ is the result of applying
  \Cref{lemma:combinematchings} with $M_a = M_{i-1}'$, $M_b=M_i$, $S=U_i$ and
  $k=2/\eps$.

	\myparagraph{\boldmath$M_1',\dots,M_t'$ satisfy the required properties:} The matching $M_1$ satisfies all the
  properties required for $M_1'$ by its definition.	
	Assume that $M_j'$ satisfies the required properties for $j< i$. 	
	We then need to show that the matching $M_c$ in the inductive construction of $M_i'$ satisfies all the
  properties required for matching $M_i'$. 
	
	\textit{For $j<i$, $M_c$ matches at least a
  $(1-\eps)\cdot\frac{\delta_j}{\actualmaxdegree+1}$-fraction of the nodes in
  $U_j$:} We already know that $M_{i-1}'$ matches at least a
  $(1-\eps)\cdot\frac{\delta_j}{\actualmaxdegree+1}$-fraction of the nodes in
  $U_j$. 
		This is also true for matching $M_c$ because by \Cref{lemma:combinematchings}, $|M_c|\geq |M_{i-1}'|$ and for every node outside $U_i$ 
  that is matched by $M_{i-1}'$ and not matched by $M_c$, there is a node in $U_i$ (and thus also in $U_j$) that is matched by $M_c$ and not by $M_{i-1}'$.
	
	\textit{$M_c$  matches a $(1-\eps)\cdot\frac{\delta_i}{\actualmaxdegree+1}$-fraction of the
  nodes in $U_i$:} As $M_b=M_i$ matches a
  $(1-\eps/2)\cdot\frac{\delta_i}{\actualmaxdegree+1}$-fraction of the nodes in
  $U_i$ and because $k=2/\eps$, it follows from
  \Cref{lemma:combinematchings} that $M_c$ matches a
  $(1-\eps/2)^2\cdot\frac{\delta_i}{\actualmaxdegree+1}$-fraction of the nodes
  in $U_i$. Thus, the matching $M_c$ satisfies the required properties because
  $(1-\eps/2)^2\geq 1-\eps$ . 
	
	\myparagraph{Runtime:} As mentioned above, we first construct matchings $M_1, \ldots, M_{t}$ in parallel, which has runtime $\approxmatching{\eps/2}$ time. 
	Then, we use \Cref{lemma:combinematchings} to inductively construct matchings $M'_1, \ldots, M'_t$ where each step requires $\bigO(1/\eps)$ time. The extension to a maximal matching takes $O(\log^3 n)=O(\approxmatching{\eps/2})$ rounds. In total, the time complexity is
$
	\bigO\left( t/\eps + \approxmatching{\eps/2}\right)$ .\end{proof}

\begin{proof}[Proof of \Cref{cor:degreepriorities}]
  The lemma directly follows by applying
  \Cref{lemma:prioritymatching} with parameter $\eps$ and by using the
  sets $V_1,\dots,V_{t + 1}$, where $V_j$ is the set of nodes
  with degree $\Delta- t +j-1$.
\end{proof}

\subsection{Polylogarithmic Runtime}
\label{sec:edgeColoringwithSplitting}
The goal is to color the edges of a graph with maximum degree $\Delta$ with $(1+\eps)\Delta$ colors while having only polylogarithmic round complexity. As the runtime of the algorithm in \Cref{thm:epsedgeColoring} depends polynomially on the maximum degree, it is only efficient for polylogarithmic maximum degree. We use the results on degree splitting by Ghaffari et al.~\cite{Splitting17} to transform it into an algorithm that is efficient for all degrees.

An \textit{undirected degree splitting} is a partition of the edges into two sets $A$ and $B$. Let $E(v)$ denote the edges incident to $v$. The discrepancy of a node $v$ in a splitting is defined as $|E(v)\cap A - E(v)\cap B|$. The objective in the \emph{degree splitting problem} is to obtain a small discrepancy at each node, i.e., the number of incident edges of $v$ in the set $A$ should be as close as possible to the number of incident edges in $B$. 

To obtain a polylogarithmic round edge coloring algorithm we recursively split the graph until the degrees are below a threshold degree $\Delta'=\polylog n$. Then each subgraph is colored, in parallel, with a separate set of $(1+\eps_2)\Delta'$ colors with \Cref{thm:epsedgeColoring}. This yields to an edge coloring of the original graph with $(1+\eps_1)(1+\eps_2)\Delta=(1+\eps)\Delta$ colors.
The following result was shown in \cite{Splitting17}.
\begin{theorem}{\cite{Splitting17}}\label{thm:mainSplitting}
For every $\gamma>0$, there is a deterministic $\tilde{O}\big(\gamma^{-1}\cdot\log\gamma^{-1}\cdot \log n\big)$-round distributed algorithm for  undirected degree splittings such that the discrepancy at each node $v$ of degree $d(v)$ is at most $\gamma\cdot d(v) + 4$, where $\tilde{O}$ hides $\log \log \gamma^{-1}$ factors.
 \end{theorem}
\begin{proof}[Proof of \Cref{thm:finalColoring}]
To edge-color the graph with $3/2\cdot \Delta$ colors, we use
\Cref{thm:threehalfedgeColoring}. Otherwise, we assume that
$\Delta\geq \Delta':= c \cdot \eps^{-1} \log \eps^{-1} \cdot \log n$
and we set $\gamma = \frac{\eps_1}{20\log \Delta}$ where $\eps_1=\eps/8$ and $\eps_2=\eps/4$. 
Then for $h$ recursive iterations, in each iteration, apply the splitting of \Cref{thm:mainSplitting} with parameter $\gamma$ to each of the parts in parallel, until we reach parts with maximum degree at most $\Delta'$. 
Here $h$ is the largest $h$ such that $(\frac{1}{2}+\gamma)^h \geq \frac{\Delta'}{\Delta}$ which is sufficient to reduce the degree to $\Delta'$.
This way we partition $G$ in $2^{h}$ edge-disjoint graphs, each with maximum degree at most $\Delta'$. 
We can then, in parallel, edge-color each of these graphs with $(1+\eps_2)\Delta'$ colors with the edge coloring algorithm from \Cref{thm:epsedgeColoring} using a separate color palette for each subgraph. 
Now, we upper bound the number of colors. At $(*)$ we use $(\frac{1}{2}+\gamma)^h \geq \frac{\Delta'}{\Delta}$ and obtain
\begin{align*}
2^h & = \left(\frac{1+2\gamma}{1/2+\gamma}\right)^h
 \stackrel{(*)}{\leq}  \frac{\Delta}{\Delta'}(1+2\gamma)^h\leq \frac{\Delta}{\Delta'}e^{\eps_1/10}\stackrel{\eps_1<1/2}{\leq} \frac{\Delta}{\Delta'}(1+2\eps_1) .
\end{align*}
Thus, the total number of colors is less than 
$2^h\cdot (1+\eps_2)\Delta' \leq   \Delta + \eps\Delta$.  
The round complexity of all the splitting iterations is 
$h \cdot \tilde{O}\big(\frac{1}{\gamma}\cdot\log\frac{1}{\gamma}\cdot \log n\big)=\tilde{O} \left(\frac{h}{\eps}\cdot\log\Delta\cdot \log n\right)$ which is submerged in  the round complexity of the  parallel invocations of  \Cref{thm:epsedgeColoring}
\begin{align*}
  \Delta'\cdot\tilde{O} \left(\frac{\log n}{\eps^{2}} + \approxmatchingx{\eps}{\Delta'} \right) 
 & =  \Delta'\cdot \tilde{O}\left( \frac{\log n}{\eps^2} + \frac{1}{\eps^2}+ \frac{1}{\eps} \cdot \match{1/\eps}{\Delta^{O(1/\eps)}}\cdot \log n +\log^3 n\right)  \\ 
 & =  \Delta' \cdot\tilde{O}\left( \frac{\log n}{\eps^2} + \frac{1}{\eps} \cdot \left(\frac{1}{\eps^7} \right) \log n \cdot \log n \cdot  \log^4 \Delta \cdot \log n +\log^3 n\right)  \\ 
 & = \tilde{O} \left( \frac{1}{\eps^{9}} \log^3 n \cdot \log^4 \Delta \cdot \log \frac{1}{\eps} \right) \ .
\end{align*}
In the above, we used that a result by Ghaffari et al.~\cite{newHypergraphMatching} that shows that
\[
	\match{r}{\Gamma}=O\big(r^2 \cdot\log(n\Gamma)\cdot\log n\cdot\log^4\Gamma\big) \ .
\]
The $\log \log$ factors in $\gamma^{-1}$ that are hidden by $\tilde{O}$ are submerged in the final polylog runtime as the condition on $\Delta$ in the lemma statement shows that $\epsilon\geq c/\Delta.$

A similar for calculation for \Cref{thm:threehalfedgeColoring} and $\Delta\leq \Delta'$ leads to a runtime of 
$\tilde{O}\big(\Delta^{9}\log^5 n \log^4\Delta\big)$ .
\end{proof}
\begin{remark}
For constant $\eps$, the runtime of \Cref{thm:finalColoring} remains polylogarithmic with the upper bound on $\match{r}{\Gamma}$ by Fischer et al.~\cite{FOCS17}.
\end{remark}
\subsection{\texorpdfstring{Going Below $\Delta+\sqrt{\Delta}$ Colors}{Going Below Delta+sqrt Delta Colors}}
\label{sec:goingBelow}
\Cref{thm:edgeColoringLastStep} does not have polylogarithmic runtime; however, it uses fewer colors than the known randomized methods. 
\begin{proof}[Proof of \Cref{thm:edgeColoringLastStep}]
If $\Delta \in \bigO(\log n)$, we can obtain a $\bigO(\log n)$ coloring in $\bigO(\Delta^9 \cdot \log^{\bigO(1)} n )= \log^{\bigO(1)} n$ rounds by \Cref{thm:finalColoring}. 
Thus, assume that $\Delta \in \omega( \log n )$.
Let $c$ be the constant in the lower bound on $\Delta$ in \Cref{thm:epsedgeColoring}. 
For $\eps = c\Delta^{-1}\cdot\log n \log \left( 2 + \frac{\Delta}{\log n} \right)$, we obtain that $\Delta \geq c\log n\frac{\log 1/\eps}{\eps}$. 
Then applying  \Cref{thm:epsedgeColoring} yields $\Delta+\eps\Delta=\Delta+c\log n \log \left(2 + \frac{\Delta}{\log n} \right)$ colors. 
By Ghaffari et al.~\cite{newHypergraphMatching}, we know that $\match{r}{\Gamma}=O\big(r^2\cdot\log(n\Gamma)\cdot\log n\cdot\log^4\Gamma\big)$ and thereby, the round complexity can be upper bounded by  
\begin{align*}
	& \bigO\left(\Delta^3+\Delta \cdot \approxmatching{\bigO(1/\eps)}\right) + \log^{\bigO(1)} n & \hspace{8.5cm}
	\end{align*}
	\begin{align*}
	\hspace{6cm} &	= \bigO\left( \Delta^3+ \frac{\Delta \log n}{\eps} \match{\bigO\left( 1/\eps \right)}{\Delta^{\bigO(1/\eps)}} \right) + \log^{\bigO(1)} n  \\
	& = \bigO\left( \Delta^3 + \Delta \cdot \frac{\Delta^8}{\log^8 n} \log^3 n \log^4 \Delta + \log^{\bigO(1)} n \right) \\
	& =  \bigO\left( \Delta^{9} + \log^{\bigO(1)} n \right) . \qedhere
\end{align*}
\end{proof}


\section{\texorpdfstring{\boldmath Deterministic Distributed  $3\Delta/2$-Edge-Coloring}{Deterministic Distributed  (3/2 Delta)-Edge-Coloring}}
\label{sec:simpleEdgeColoring}
Let $G=(V,E)$ be a graph with maximum degree at most $\Delta$.
There are two crucial points to compute a $3\Delta/2$-edge-coloring of $G$: First, we can extract a so called $(3)$-graph $H=(V,F)$ from $G$ and guarantee that the maximum degree of the graph $(V,E - F)$ is at most $\Delta-2$. 
Second, we can efficiently $3$-color $H$. Repeating these two steps $\Delta/2$ times with a fresh set of colors for each extracted subgraph yields the result.

The extraction of $(3)$-graphs relies on the computation of a sequence of five maximal and maximum matchings of subgraphs of $G$. 
While computing a maximum matching is a global problem in general bipartite graphs, we only compute maximum matchings in a special class of bipartite graphs. In these graphs the maximum degree on one side is smaller than the minimum degree on the other side. (cf. \Cref{lemma:maximumMatching}). 

 We begin with the definition of $(3)$-graphs.
\begin{definition}
A \textit{$(3)$-graph} is a graph with maximum degree $3$ where no two degree-$3$ vertices are adjacent.

\end{definition}
In it was shown by Fournier\cite{Fournier73} -- via a global argument -- that $(3)$-graphs admit $3$-edge-colorings. 
As the line graph of a $(3)$-graphs has maximum degree $3$ the node coloring algorithm by Panconesi and Srinivasan \cite[\mbox{Theorem 3}]{Panconesi1995} colors the line graph with  $3$ colors, i.e., we can  $3$-edge-color $(3)$-graphs with a distributed algorithm.
\begin{lemma}
\label{lemma:threeColoring}
$(3)$-graphs can be edge-colored in time $O(\log^3 n)$ with $3$ colors.
\end{lemma}

We continue with computing maximum matchings in certain bipartite graphs. 
More precisely, we can efficiently compute maximum matchings if the minimum degree on one side of the graph is larger than the maximum degree on the other side. 
We first bound the length of augmenting paths in such graphs.
Notice that in this section, we only consider unweighted graphs and therefore, we consider the standard version of an ``augmenting path'', i.e., a path with every second edge in a matching and non-matched endpoints.
\begin{lemma}
\label{lem:bipartitepathlength}
Assume we are given a bipartite graph $B = (U, V, E)$, where all nodes in $U$ have degree at least $d$ and all nodes in $V$ have degree at most $f<d$. Further assume that we are given a matching $M$ of $B$ and let $S \subseteq U$ be the set of nodes in $U$ that are not matched by $M$. 
Then, for each node $s$ in $S$, there exists an augmenting path of length at most $O(d \log n)$.
\end{lemma}
\begin{proof}
Consider the directed graph $B'$ that is obtained from $B$ by orienting all $M$-edges from $V$ to $U$ and all other edges from $U$ to $V$. A node $s \in S$ has an augmenting path of length $L$ iff in $B'$, there is a directed path of length $L$ from $s$ to an unmatched node in $V$. 

Fix some unmatched node $s$ and assume that there is no augmenting path of length at most $L$. For $i\geq 0$, let $X_i$ be the set of nodes in $U$ that are within distance at most $2i$ from $s$ in the directed graph $B'$. Further, let $Y_i\subseteq V$ be the set of nodes that are within distance at most $2i+1$ from $s$ in $B'$. 
As long as $2i+1\leq L$, we have $|X_{i+1}| = |Y_i|+1$ because we have the unmatched $s \in X_{i+1}$ and in addition, every node in $Y_i$ has exactly one matched out-neighbor in $X_{i+1}$. Because in the original bipartite graph $B$, every node in $U$ has at least $d$ neighbors in $V$ and every node in $V$ has at most $f\leq d-1$ neighbors in $U$, we also have 
$$|Y_i| \geq d/f\cdot |X_i|\geq d/(d-1)\cdot |X_i|~.$$
Combining both facts, we thus have 
$|X_{i+1}| \geq d/(d-1)\cdot|X_i|\text{~, if ~} 2i+1\leq L~.$
 If $L\geq c d\log n$ for a sufficiently large constant $c$, this leads to a contradiction.
\end{proof}

\begin{lemma}
\label{lemma:maximumMatching}
There is an $O(d\cdot \log n\cdot \match{d\log n}{d^{O(d\log n)}})$ time distributed deterministic algorithm that computes a maximum matching in bipartite graphs $G=(U\dcup V, E)$ with $$d:=\min_{u\in U}deg(u)>f:=\max_{v\in V}deg(v)~.$$
\end{lemma}
\begin{proof}
Augmenting paths can be seen as hyperedges of a hypergraph defined on the node set of $G$. We consider augmenting paths with length up to $\ell=O(d\log n)$, i.e, the hypergraph has rank at most $\ell$ and degree at most $d^{O(\ell)}$. With this hypergraph correspondence at hand we iteratively compute maximal independent sets of shortest augmenting paths in time $\match{\ell}{d^{O(\ell)}}$~\cite{newHypergraphMatching}. A classic result \cite{Hopcroft73}  shows that augmenting along a maximal independent set of augmenting paths increases the length of the shortest augmenting path. Thus, using \Cref{lem:bipartitepathlength} we only need $O(d\log n)$ iterations of computing maximal independent sets of augmenting paths until we have computed a maximum matching.
\end{proof}
Note that a maximum matching of a graph as in \Cref{lemma:maximumMatching} matches all nodes in $U$.

\begin{lemma}[Extracing $(3)$-graphs]
\label{lemma:extraction}
Let $G=(V,E)$ be a graph with maximum degree at most $\Delta\geq 3$. There is a deterministic distributed algorithm with time complexity $$O(\Delta \cdot \log n \cdot \match{\Delta\log n}{\Delta^{O(\Delta\log n)}}+\log^3 n)$$ that computes an edge set $F\subseteq E$ such that $H=(V,F)$ is a $(3)$-graph and the maximum degree of the graph $(V,E - F)$ is at most $\Delta-2$. 
\end{lemma}
\myparagraph{Algorithm:}
The algorithm consists of the following five steps in which we compute maximal and maximum matchings of subgraphs of $G$.
Note that maximal matchings can be computed in $O(\log^3n)$, \cite{Fischer17}.

\begin{DenseEnumerate}
\item  Let $V_\Delta := \{v \in S : \deg_G(v) = \Delta\}$ be the nodes with maximum degree in $G$. Compute a maximal matching $M_1$ of $G[V_\Delta]$ and let $G_1 := (V, E - M_1)$ be the graph that remains if we remove $M_1$.

\item Let $V_{1,\Delta} := \{v \in S : \deg_{G_1}(v) = \Delta\}$ be the maximum degree nodes of $G$ that have not been hit by $M_1$. 
Let $B_1$ be the bipartite subgraph of $G_1$ that is spanned by $V_{1,\Delta}$ and $V  - V_{1,\Delta}$. 

Compute a maximum matching $M_2$ of $B_1$ via \Cref{lemma:maximumMatching} and let $G_2:=(V, E  - M_1  - M_2)$ be the graph that remains if we remove $M_2$ from $M_1$.

\item  Let $V_{2,\Delta-1} :=\{v \in S : \deg_{G_2}(v) = \Delta-1\}$ be the nodes that have degree $\Delta-1$ in $G_2$. 

Compute a maximal matching $M_3$ of $G_2[V_{2,\Delta-1}]$ and define $G_3 := (V, E  - M_1   - M_2  - M_3)$.

\item Let $V_{3,\Delta-1} := \{v \in S : \deg_{G_3}(v) = \Delta-1\}$ and let $B_3$ be the bipartite subgraph of $G_3$ that is spanned by $V_{3,\Delta-1}$ and $V  - V_{3,\Delta-1}$. 

Compute a maximum matching $M_4$ of $B_3$ via \Cref{lemma:maximumMatching} and define $H' := (V, M_1 \cup M_2 \cup M_3 \cup M_4)$.

\item Compute a maximal matching $M'$ of degree $3$ nodes in $H'$ and let $H := (V, M_1 \cup M_2 \cup M_3 \cup M_4  - M')$.
\end{DenseEnumerate}
\begin{proof}[Proof of \Cref{lemma:extraction}]

We first prove that we can apply \Cref{lemma:maximumMatching} in step two and four.

\myparagraph{Step two:}
$V_{1,\Delta}$ is an independent set in $G_1$ because otherwise $M_1$ would not be maximal.
Thus, we have that each of the $\Delta$ edges adjacent in $G_1$ to  node $v\in V_{1,\Delta}$ is an edge in $B_1$. That is, $$d:=\min_{v\in V_{1,\Delta}} \deg_{B_1}(v)= \Delta~.$$
  By definition every node in  $V - V_{1,\Delta}$ has degree at most $\Delta-1$ in $G_1$ and also in $B_1$. 
Thus, we have
$$f:=\max_{v\in (V - V_{1,\Delta})}\deg_{B_1}(v)\leq \Delta-1~.$$

\myparagraph{Step four:}
$V_{3,\Delta-1}$ is an independent set in $G_3$ because otherwise $M_3$ would not be maximal.
Thus, we have that each of the $\Delta-1$ edges adjacent in $G_3$ to node $v\in V_{1,\Delta}$ is an edge in $B_3$. That is, 
$$d:=\min_{v\in V_{3,\Delta-1}} \deg_{B_3}(v)= \Delta-1~.$$  By definition every node in  $V - V_{3,\Delta-1}$ has degree at most $\Delta-2$ in $G_3$ and also in $B_3$. Thus we have $$f:=\max_{v\in (V - V_{3,\Delta-1})}\deg_{B_3}(v)\leq \Delta-2~.$$

\medskip

It is sufficient to show the following three properties:
$a)$ $H$ is a $(3)$-graph. $b)$ Every node with degree $\Delta$ in $G$ has at least degree two in $H$. $c)$ Every node with degree $\Delta-1$ in $G$ has at least degree one in $H$. 

\myparagraph{Property $a)$:} The nodes of $H$ have at most degree three: As $M_1$ and $M_2$ are matchings  each node has degree at most two in $M_1 \cup M_2$. In the third step only nodes with degree one in $M_1\cup M_2$ get at most one additional edge, i.e., every node has at most two adjacent edges in $M_1 \cup M_2 \cup M_3$. In the fourth step we add a single matching, i.e., every node has at most three adjacent edges in $M_1 \cup M_2 \cup M_3\cup M_4$. In step $5$ we only remove edges so the degree bound still holds.

The nodes of degree $3$ in $H$ form an independent set because we remove a maximal matching between all degree three nodes in step $5$.

\myparagraph{Property $b)$:} 
A node with degree $\Delta$ in $G$ is hit at least once by $M_1\cup M_2$ because every node with degree $\Delta$ in $G$ which is not hit by $M_1$ is for sure hit by the maximum matching $M_2$. If it was only hit once by $M_1\cup M_2$ then it will be hit again at least once by $M_3\cup M_4$. Furthermore, every node with degree at least two in $H'$ has degree at least two in $H$ as well.

\myparagraph{Property $c)$:}
A node with degree $\Delta-1$ in $G$ is, if not hit by $M_1\cup M_2 \cup M_3$, for sure hit by $M_4$. 
Furthermore, any node which has degree at least one in $H'$ has degree at least one in $H$.
\end{proof}

\begin{restatable}{lemma}{thmsecondedgecoloring}
\label{thm:threehalfedgeColoring}
There is a deterministic distributed algorithm that computes a $3\Delta/2$-edge-coloring of any $n$-node graph with maximum degree at most $\Delta$ in $O\big(\Delta^2 \cdot \log n\cdot \match{\Theta(\Delta\cdot \log n)}{\Delta^{O(\Delta\log n)}}\big)$ rounds.
\end{restatable}
\begin{proof}
Let $G_0:=G=(V,E)$. We iteratively extract $k=\lfloor\frac{\Delta-1}{2}\rfloor$ $(3)$-graphs $H_1,H_2,\ldots, H_k$ where the $(3)$-graph $H_i=(V,F_i)$ is obtained by applying \Cref{lemma:extraction} to $G_{i-1}$.
We set $G_{i}:=(V,E_{i})=(V,E_{i-1}  - F_i)$.  

Then the edge sets $F_1,\ldots, F_k, E_k$ form a partition of $E$ and we color each of the sets with a separate set of colors as follows.
Use \Cref{lemma:threeColoring} to edge-color each $H_i, i=1,\ldots,k$ with a fresh set of three colors. 
If $\Delta$ is even, the maximum degree of $G_{k}$ is at most two and we can color $G_k$ with three colors with the method by Cole and Vishkin \cite{cole86}  in time $O(\logstar n)$.
If $\Delta$ is odd the maximum degree of $G_{k}$ is at most one and we can edge-color it with a single color in a single round.
Altogether we use $3\Delta/2$ colors and the time complexity follows from the $k$ invocations of \Cref{lemma:threeColoring} and \Cref{lemma:extraction}.
\end{proof}


\section{Approximate Weighted Matching}
\label{sec:weightedmatching}
In this section, we show that the \CREWPRAM algorithm by Hougardy and Vinkemeier~\cite{Hougardy2006}, which approximates a weighted maximum matching, can be adapted to the distributed setting.

\myparagraph{Augmentations.} 
Let $G = (V, E)$ be a graph with positive edge weights $w:E\to \mathbb{R}^+$.
Let $S$ and $M$ be matchings in $G$ and consider the symmetric difference of $S$ and $M$.
We call $S$ an augmenting path (resp. cycle) of $M$ if the symmetric difference is a path (resp. cycle) and $S$ is a matching in $E - M$.
Let $M(S)$ denote all the edges in $M$ that have a node in common with $S$ and $w(S)$ the sum of edge weights in $S$ (resp. $M$).
For simplicity, we refer to both augmenting paths and cycles simply as \emph{augmentations}.
The number of edges in an augmentation is referred to as its \emph{length} and $w(S) - w(M(S))$ to as the \emph{gain} of $S$.
Notice that the edges of matching $M$ are not counted into the length of the augmentation.
From here on, we only consider augmentations of length at most $\ell= 2/\eps$ and with positive gain.

\myparagraph{Ranks.}
We divide the augmentations into \emph{ranks} according to their gain.
Let $\wmin$ and  $w_{\textrm{max}}$ be the minimum and the maximum edge weights of $G$, respectively.
We assume that both of these values are known to the protocol. 
To later obtain a logarithmic number of ranks, we assume that $w_{\textrm{max}} / w_{\textrm{min}}$ is polynomial in $n$.

For an augmentation $S$ with gain $g(S)$, the rank $r(S)$ is defined as
\begin{enumerate}
	\item $r(S) = 0$, if $g(S) \leq \frac{\wmin}{\ell \cdot n}$.
	\item $r(S) = i$, where $i$ is the smallest such that $g(S) \leq 2^i \cdot \frac{\wmin}{\ell \cdot n}$.
\end{enumerate}
Notice that according to this definition, for any two augmentations of the same rank $i > 0$, the gain is within a factor of $2$.
However, in the case of rank $0$, this does not necessarily hold.

\myparagraph{The case of small $w(M^*)$.} 
For a technical reason, we perform a preprocessing step of $\bigO(1/\eps)$ rounds, where the nodes check if $w(M^*) < (1/\eps) \wmin$.
Notice that this can be the case only if the diameter of the graph is at most $2/\eps$.
In this case, we can simply choose a maximum matching as our output.
Therefore, we can assume for the rest of the section that $w(M^*) \geq (1/\eps) \wmin$.	
	
To obtain a $(1-\eps)$-approximate weighted matching, our algorithm begins with an empty matching and augments the matching $O(1/\eps)$ times with a maximal set of independent augmentations.
To compute such a set, we construct  hypergraphs $H_1,\ldots, H_{r_{\textrm{max}}}$ on the node set $V$. The set of edges of $H_i$ corresponds to the augmenting paths and cycles of length at most $\ell$ with rank $i$ with regard to the current matching.	
 Then, for $i = r_{\textrm{max}}, \ldots, 1$, find a maximal matching in $H_i$.
Before proceeding to $H_{i - 1}$, remove the matched  nodes in $H_i$ from  all $H_j, j<i$.
Notice that the algorithm does not update the ranks of the augmentations while we iterate through the hypergraphs.
The union of the hypergraph matchings corresponds to the set of augmentations which we use to augment the overall matching. 	

\medskip

\noindent\textbf{Differences to the approach by Hougardy and Vinkemeier~\cite{Hougardy2006}:} Much of the above is along similar ideas in \cite{Hougardy2006}. However, the hypergraphs which we construct consists of augmentations which are formed by paths and cycles whereas the corresponding part in \cite{Hougardy2006} contains arbitrary augmentations, e.g., unions of paths which are far away from each other in the network graph. 
In the \LOCAL model it is not possible to construct the hypergraphs efficiently if we allow those arbitrary augmentations.
Secondly, our rank definition differs slightly from the one in \cite{Hougardy2006} (we use $\wmin$ instead of $\gmax$). Thirdly, we handle the case of small $w(M^*)$ separately because due to the altered rank definition we will later use that $w(M^*) \geq (1/\eps) \wmin$ holds.
Due to these changes we cannot use their analysis of the algorithm as a blackbox. However, almost every line of the following analysis is similar to the proof by Hougardy and Vinkemeier.

\newcommand{\OPT}{\textrm{OPT}}
\newcommand{\ALG}{\textrm{ALG}}
Let us consider one iteration of our protocol, i.e., augmenting a matching $M$ with the union of ranked augmenting paths and cycles.
Let $M^*$ be a maximum weight matching in $G$ and consider the symmetric difference of $M$ and $M^*$.
Let $C$ be a maximal cycle in this symmetric difference and let $C^* = C \cap M^*$.
For the following definition, assume that the cycle is consistently oriented.
We consider a multiset \OPT{} that contains $\ell$ copies of $C^*$ if $|C^*| \leq \ell$ and otherwise, for every edge $e \in C^*$, we insert an augmenting path of length $\ell$ that contains $e$ and the $\ell - 1$ edges following $e$ (according to the consistent orientation).
In the case that $C$ is a path, we simply imagine that the endpoints are connected and handle $C$ as in the case for cycles.

Now, by definition, any edge in $M^* - M$ is contained in at least $\ell$ augmentations in \OPT{}.
Consider an edge $e \in M - M^*$, that is part of a short path or cycle in the symmetric difference of $M$ and $M^*$.
For such an edge, there are $\ell$ augmentations $S \in \OPT$ such that $e$ is contained in $M(S)$ (recall that in $M(S)$ are the edges of $M$ that have a common node with $S$).
For the case that $e$ is part of a long cycle or path, it can be the case that there are $\ell + 1$ augmentations $S \in \OPT$ such that $e$ is contained in $M(S)$.
In a cycle, for example, $e$ is connected to an augmentation in $\OPT$ that starts with an edge in front of $e$ and to the $\ell$ augmentations that contain the edge before $e$ (again, according to the consistent orientation).
We get that
\begin{equation}
	\label{eq:optbound}
	\sum_{S \in \OPT} g(S) \geq \ell \cdot w(M^*) - (\ell + 1) \cdot w(M) \ .
\end{equation}
Recall, that we assumed  that $w(M^*) \geq (1/\eps) \wmin$.
Thus, if $w(M^*) - w(M) < \wmin$, it holds that 
\[
	w(M) > w(M^*) - \wmin \geq w(M^*) - w(M^*) \eps
\]
and then the matching $M$ is already a $(1 - \eps)$-approximation.
We can therefore assume that $$\wmin \leq w(M^*) - w(M)~.$$ 
Given the construction of \OPT{}, we get that the number of augmentations in \OPT{} is bounded by $\ell \cdot n$ and by identifying $H_0$ with the set of corresponding augmentations, it follows that,
\begin{equation}
	\label{eq:rank0}
	\sum_{S \in \OPT \cap H_0} g(S) \leq \ell \cdot n \cdot \frac{\wmin}{n \cdot \ell} \leq w(M^*) - w(M) \ .
\end{equation}
In the next two lemmas we prove the approximation guarantee and finally, \cref{lemma:distrweightedmatching} follows by bounding the runtime.

\newcommand{\OPTO}{\ensuremath{\textrm{OPT}_{>0}}}
\begin{lemma}
\label{lemma:improve}
	Let $M$ be a matching. Applying one iteration of our matching augmentations results in a matching $M'$ such that 
	\[
		w(M') \geq w(M) + \frac{1}{4\ell} \cdot \left( \frac{\ell - 1}{\ell} ( w(M^*) - w(M) ) \right) \ . 
	\]
\end{lemma}
\begin{proof}
	Consider some $S \in \OPT$ and assume $r(S) = i > 0$.
	Let \ALG{} be the set of augmentations, i.e., the hypergraph matching, computed by our algorithm.
	Since our algorithm picks maximal sets of hyperedges with decreasing ranks, we get that \ALG{} contains an augmentation, with rank $i$ or higher and that has a node common with augmentation $S$.
	If we assign $S$ to such a node, we get due to the definition of \OPT{} that at most $\ell$ augmentations are assigned per node.
	Recall that the length of an augmentation $S$, that is a matching, corresponds to the number of edges in $S$. 
	Thus, the edges of an augmentation of length at most $\ell$ are incident on at most $2\ell$ nodes.
	Let $\OPTO \subseteq \OPT$ be the augmentations of rank higher than $0$.
	It follows that
	\begin{align}\label{eqn:ALG}
		g(\ALG) \geq \frac{1}{2} \frac{1}{\ell} \frac{1}{2\ell} \sum_{S \in \OPTO} g(S) \geq \frac{1}{4\ell^2} \sum_{S \in \OPTO} g(S) \ .
	\end{align}
	We can use \cref{eq:optbound,eq:rank0} to obtain 
	\begin{align*}
		\sum_{S \in \OPTO} g(S) & = \sum_{S \in \OPT} g(S) - \sum_{S \in \OPT \cap H_0} g(S) \\
			 & \stackrel{\cref{eq:optbound}}{\geq} \ell \cdot w(M^*) - (\ell + 1) \cdot w(M) - \sum_{S \in \OPT \cap H_0} g(S) \\ 
			 & \stackrel{\cref{eq:rank0}}{\geq} \ell \cdot w(M^*) - (\ell + 1) \cdot w(M) - w(M^*) - w(M) = (\ell - 1) \cdot w(M^*) - \ell \cdot w(M) \ .
	\end{align*}

	Combining with \Cref{eqn:ALG}, we get that 
	\begin{align*}
		g(\ALG) & \geq \frac{1}{4\ell^2} \cdot \left( (\ell - 1) \cdot w(M^*) - \ell \cdot w(M) \right) = \frac{1}{4\ell^2} \cdot \ell \cdot \left( \frac{\ell - 1}{\ell} ( w(M^*) - w(M) ) \right) \\
				& = \frac{1}{4\ell} \cdot \left( \frac{\ell - 1}{\ell} ( w(M^*) - w(M) ) \right) \ .
	\end{align*}
	Finally, the result follows from the fact that $g(\ALG) = w(M') - w(M)$, 
\end{proof}

\begin{lemma}
	\label{lemma:approx}
	Let $G$ be an edge-weighted graph and assume that $\wmin$ is known to the nodes. Then, for every $\eps > 0$, there is a deterministic distributed algorithm that finds a matching $M$ such that $$w(M) \geq (1 - \eps)w(M^*)~.$$ 
\end{lemma}
\begin{proof}
	Let $M_0$ be an empty matching and $M_i$ a matching obtained by applying one iteration of the augmentations from the hypergraph matching procedure.
	By \cref{lemma:improve}, we get that 
	\[
		w(M_{i + 1}) \geq w(M_i) + \frac{1}{4\ell} \cdot \left( \frac{\ell - 1}{\ell} w(M^*) - w(M_i) \right) \ .
	\]
	Solving this recurrence (c.f.~\cite[Proof of Theorem 2]{Hougardy2006}) yields that for some $k \in \bigO(1/\eps)$ we get that $w(M_k) \geq (1 - \eps) w(M^*)$.
\end{proof}

\begin{proof}[Proof of \cref{lemma:distrweightedmatching}]
Constructing the hypergraphs $H_1, \ldots, H_{r_{\textrm{max}}}$ for a single augmentation can be done in parallel in $\bigO(\ell)$ rounds.
By definition, we get that the number of nodes per hyperedge (i.e., the hyperedge rank) in $H$ is bounded from above by $\ell$.
Notice that one round of communication in this hypergraph can take up to $\ell$ rounds, since the nodes adjacent in the hypergraph might be up to $\ell$ hops away in the underlying communication graph. 
Given that $\wmax / \wmin \in n^{\bigO(1)}$ and $\ell \leq n$, we get that $r_{\textrm{max}} \in \bigO(\log n)$. The degree of each hypergraph is upper bounded by $\Delta^{O(\ell)}$.
	Therefore, iterating through the hypergraphs and finding a maximal matching in each of them can be done in  
	\[
		\bigO\left( \ell \cdot \match{\ell}{\Delta^{O(\ell)}} \cdot\log n\right)
	\]
	rounds where $\match{r}{\Gamma}$ is the time of an algorithm which solves the maximal matching algorithm on a hypergraph with $n$ nodes, maximum degree at most $\Gamma$ and rank $r$.  		
	
	Once the augmentation has been computed, they can be applied in time $O(\ell) = \bigO(1/\eps)$.  	
	 As we iteratively compute $\bigO(1/\eps)$ augmentations the total round complexity is
\[
	\bigO\left( (1/\eps) \cdot ((1/\eps) + \match{\ell}{\Delta^{O(\ell)}}\cdot \log n) \right) = \bigO\left(1/\eps^2+ (1/\eps) \cdot \match{\ell}{\Delta^{O(\ell)}}\cdot \log n      \right) .
\]
	The result follows by applying \cref{lemma:approx} and completing the matching into a maximal matching in time $\bigO(\log^3 n)$ with the algorithm by Fischer~\cite{Fischer17}.
\end{proof}


\section{Derandomizing Randomized Edge-Coloring Algorithms}
\label{sec:derandomized}

In \cite{newHypergraphMatching}, Ghaffari, Harris, and Kuhn have
recently developed generic methods to derandomize distributed
algorithm in the \LOCAL model. We next describe how these techniques
can be applied in a blackbox way to existing randomized in order to
get deterministic polylogarithmic time edge-coloring algorithms. A
special case of Theorem 1.1 in \cite{newHypergraphMatching}
immediately gives the following generic distributed derandomization
result. 

\begin{lemma}[Special case of Theorem 1.1 in
  \cite{newHypergraphMatching}]\label{lemma:derandomization}
  Let $G=(V,E)$ be an $n$-node graph with maximum degree $\Delta$. Any
  $r$-round randomized \LOCAL algorithm for a locally checkable
  problem on $G$ can be transformed to a deterministic \LOCAL
  algorithm on $G$ with time complexity
  $O\big(r\cdot(\Delta^{O(r)}+\log^* n)\big)$.
\end{lemma}

A problem is called locally checkable if the correctness of a solution
can be checked in $O(1)$ rounds in the \LOCAL model: If the solution
is correct, all nodes output ``yes'', if the solution is not correct,
at least one node must output ``no'' (cf.\ \cite{fraigniaud13}). We
apply the above lemma to the randomized edge-coloring algorithm of
\cite{Pettie17}, where the following statement is proven.

\begin{lemma}[Theorem 3.1 in \cite{Pettie17}]\label{lemma:randomizedalg}
  Let $G=(V,E)$ be an $n$-node graph with maximum degree $\Delta$ and
  let $\eps=\omega\big(\frac{\log^{2.5}\Delta}{\sqrt{\Delta}}\big)$ be
  a function of $\Delta$. If $\Delta> \Delta_\eps$ is sufficiently
  large, an $(1+\eps)\Delta$-edge coloring of $G$ can be computed in
  the time required to solve 
  \begin{itemize}
  \item[a)] $O(\log(1/\eps))$ instances of a symmetric Lov\'asz Local
    Lemma (LLL)
    problem on an $n$-vertex dependency graph with maximum degree
    $\Delta^{O(1)}$ and where each of the bad events occurs with
    probability at most
    $\exp\big(-\eps^2\Delta/\log^{4+o(1)}\Delta\big)$, as well as
  \item[b)] one instance of an $O(\eps\Delta)$-edge coloring problem on a
    graph of maximum degree $O(\eps\Delta)$.
  \end{itemize}
\end{lemma}

Combining the two lemmas with the randomized distributed LLL algorithm
of \cite{pettie_LLLalg} and the distribouted degree splitting
algorithm of \cite{Splitting17}, we obtain the following theorem.

\begin{theorem}\label{thm:derandomizedEdgeColoring}
  Let $G=(V,E)$ be an $n$-node graph with maximum degree at most $\Delta$ and
  let
  $\eps=\omega\big(\frac{\log^{2.5}\Delta}{\sqrt{\Delta}}\big)$. If
  $\Delta\geq c\cdot \frac{\log n \,\cdot\,(\log\log n)^{4+o(1)}}{\eps^2}$ for a
  sufficiently large constant $c>0$, there
  exists a deterministic distributed algorithm to compute a
  $(1+\eps)\Delta$-edge coloring of $G$ in $(\log(n) /
  \eps)^{O(1)}$ rounds.
\end{theorem}
\begin{proof}[Proof Sketch]
  We here show how to get a deterministic edge coloring algorithm
  with time complexity $\Delta^{O(1)} +O(\log^* n)$. A time complexity of
  $(\log(n) / \eps)^{O(1)}$ can then be achieved by using the degree
  splitting algorithm of \cite{Splitting17}, in the same way as we did in the
  proof of \Cref{thm:finalColoring}, by effectively reducing the maximum degree to $\Theta\big(\frac{\log n \,\cdot\,(\log\log
    n)^{4+o(1)}}{\eps^2}\big)$.

  Assume that
  $\Delta \geq c\cdot \frac{\log n \,\cdot\,(\log\log
    n)^{4+o(1)}}{\eps^2}$
  for a sufficiently large constant $c>1$. We use
  \Cref{lemma:derandomization} to derandomize the solution of each of
  the LLL instances in the randomized edge coloring algorithm of
  \Cref{lemma:randomizedalg}. The best randomized distributed LLL in
  our context is from the work of Chung, Pettie, and Su~\cite{pettie_LLLalg}. Assume that we are given
  an $n$-vertex dependency graph with maximum degree $d$ and where
  each bad event has probability at most $p$. If $epd^2<1$, the
  algorithm of \cite{pettie_LLLalg} computes a solution in
  time $O(\log(n)/\log(1/epd^2))$. In our case, we have
  $d = \Delta^{O(1)}$ and
  $p = \exp\big(-\eps^2\Delta/\log^{4+o(1)}\Delta\big)$. Hence, we have $epd^2 = exp(-\Theta(\log n))$
  and the LLL algorithm of \cite{pettie_LLLalg} thus has a constant
  time complexity. Using \Cref{lemma:derandomization}, we can therefore
  turn the randomized LLL algorithm into a deterministic distributed
  algorithm with time complexity $\Delta^{O(1)} + O(\log^* n)$. In the last step of
  the randomized edge coloring algorithm of \cite{Pettie17}, one
  needs to compute an $O(\eps\Delta)$-edge coloring on a graph of
  maximum degree $O(\eps\Delta)$. This can be done deterministically
  in time $O(\Delta + \log^* n)$ by using the
  $(\Delta+1)$-vertex coloring algorithm of \cite{BEK15}.
\end{proof}


\clearpage
\bibliographystyle{alpha}
\bibliography{references}

\end{document}